\documentclass{patmorin}
\usepackage{pat}
\usepackage[utf8]{inputenc}
\usepackage[noend]{algorithmic}
\usepackage{subcaption}
\usepackage{enumitem}
\usepackage{url}
\usepackage{comment}

\newtheorem{innercustomthm}{Theorem}
\newenvironment{customthm}[1]
  {\renewcommand\theinnercustomthm{#1}\innercustomthm}
  {\endinnercustomthm}

\title{\MakeUppercase{Encoding Arguments}\thanks{This research is partially funded by NSERC.
W. Mulzer is supported by DFG Grants 3501/1 and 3501/2.}}
\author{Pat Morin, Wolfgang Mulzer, and Tommy Reddad}
\date{}

\begin{document}
\begin{titlepage}
\maketitle

\begin{abstract}
  \setlength{\baselineskip}{15.84pt}

  Many proofs in discrete mathematics and theoretical computer  science
  are based on the probabilistic method. To prove the existence of a
  good object, we pick a random object and show that it is bad with low
  probability. This method is effective, but the underlying probabilistic
  machinery can be daunting. ``Encoding arguments'' provide an alternative
  presentation in which probabilistic reasoning is encapsulated in a
  ``uniform encoding lemma''. This lemma provides an upper bound on the
  probability of an event using the fact that a uniformly random choice
  from a set of size $n$ cannot be encoded with fewer than $\log_2 n$
  bits on average. With the lemma, the argument reduces to devising an
  encoding where bad objects have short codewords.

  In this expository article, we describe the basic method and
  provide a simple tutorial on how to use it. After that, we survey
  many applications to classic problems from discrete mathematics
  and computer science. We also give a generalization for the case
  of non-uniform distributions, as well as a rigorous justification
  for the use of non-integer codeword lengths in encoding arguments.
  These latter two results allow encoding arguments to be applied more
  widely and to produce tighter results.
\end{abstract}


\end{titlepage}
\pagenumbering{roman}
\tableofcontents
\newpage
\pagenumbering{arabic}

\section{Introduction}
\setlength{\baselineskip}{15.84pt}

There is no doubt that probability theory plays a fundamental role in
computer science. Often, the fastest and simplest solutions to important
algorithmic and data structuring problems are randomized~\cite{mitzenmacher.upfal:probability,motwani.raghavan:randomized}; 
average-case analysis of algorithms relies entirely on tools from probability 
theory~\cite{flajolet.sedgewick:aofa}; 
many difficult combinatorial questions have strikingly simple
solutions using probabilistic arguments~\cite{alon:probabilistic}.

Unfortunately, many of these beautiful results 
present a challenge to
most computer scientists, as they rely on
advanced mathematical concepts.
For instance, the 2013 edition of ACM/IEEE Curriculum
Guidelines for Undergraduate Degree Programs in Computer Science does
not require a full course in probability theory
\cite[Page~50]{computing-curricula:computer}. Instead, the report
recommends a total of 6 Tier-1 hours and 2 Tier-2 hours on
discrete probability, as part of the discrete structures curriculum
\cite[Page~77]{computing-curricula:computer}.

``Encoding arguments'' offer a more elementary 
approach to presenting these results.
We transform the
task of upper-bounding the probability of a specific event,
$\mathcal{E}$, into the task of devising a code for the set of
elementary events in $\mathcal{E}$. This provides an alternative
to performing a traditional probabilistic analysis or, since we are
only concerned with finite spaces, to directly estimating the size of
$\mathcal{E}$. Just as the probabilistic method
is essentially only a sophisticated rephrasing of a counting argument
with many theoretical and intuitive advantages,
encoding arguments likewise
offer their own set of benefits,
although they are also just a glorified way of counting. More
specifically:

\begin{enumerate}
\item Except for
  applying a simple \emph{Uniform Encoding Lemma}, 
  encoding arguments are ``probability-free.''
  There is no danger of
  common mistakes such as multiplying probabilities of non-independent
  events or (equivalently) multiplying expectations.

  The actual proof of the Uniform Encoding Lemma is trivial. It only
  uses the probabilistic fact that if we have a finite set $X$
  with $r$ special elements and we pick an element from $X$ uniformly at
  random, the probability of selecting a special element
  is $r/|X|$.

\item Encoding arguments usually yield strong results;
  $\Pr\{\mathcal{E}\}$ typically decreases at least exponentially in
  the parameter of interest. Traditionally, these strong 
  results require (at least) careful calculations on probabilities of
  independent events and/or concentration
  inequalities.  This latter subject is advanced
  enough to fill entire textbooks
  \cite{boucheron.lugosi.ea:concentration,dubhashi.panconesi:concentration}.
  
\item Encoding arguments are natural for computer scientists. They
  turn a probabilistic analysis into the task of designing
  an efficient code---an algorithmic problem. Consider the following
  two problems:
  \begin{enumerate}

  \item Prove an upper-bound of $1/n^{\log n}$ on the probability that
    a random graph on $n$ vertices contains a clique of size $k=\lceil
    4\log n\rceil$.\footnote{Since we are overwhelmingly concerned with
    binary encoding, we will
    agree now that the base of logarithms in $\log x$ is $2$, except when
    explicitly stated otherwise.}

  \item Design an encoding for graphs on $n$ vertices so that any graph
    with a clique of size $k=\lceil 4\log n\rceil$ is
    encoded using at most $\binom{n}{2}-\log^2 n$ bits. (Note: Your
    encoding and decoding algorithms don't have to be efficient, just
    correct.)
  \end{enumerate}
  Many computer science undergraduates would not know where to start
  on the first problem.  Even a good student who realizes that they
  can use Boole's Inequality will still be stuck wrestling with the
  formula $\binom{n}{4\log n}2^{-\binom{4\log n}{2}}$.
\end{enumerate}

We believe that encoding arguments are an easily
accessible, yet versatile tool for solving many problems.  Most of
these arguments can be applied after learning almost no probability
theory beyond the Encoding Lemma mentioned above.

The remainder of this article is organized as follows. In
\secref{uel}, we present an elementary tutorial on the
method, including the
\emph{Uniform Encoding Lemma}, the basis of most of our
encoding arguments. In \secref{background}, we review 
more advanced mathematical tools, such as entropy,
Stirling's Approximation, and encoding schemes for natural
numbers. In \Secref{applications-i} we show how the Uniform
Encoding Lemma can be applied to
Ramsey graphs, several hashing variants, expander graphs,
analysis of binary search trees, and $k$-SAT. In
\Secref{nuel}, we introduce the Non-Uniform Encoding Lemma,
a generalization
that extends the reach of the method, This is demonstrated in
\secref{applications-ii}, where we prove the Chernoff 
bound and consider percolation and random triangle counting problems.
\Secref{el} presents an
alternative view of encoding arguments, justifying the use of
non-integer codeword lengths.  \Secref{summary} concludes the
survey.

\section{An Elementary Tutorial}
\seclabel{uel}

This section gives a simple tutorial on the basic
method.

\subsection{Basic Definitions, Prefix-free Codes and the Uniform 
Encoding Lemma}

Before we can talk about codes, we first recall
some basic definitions about bit strings.
\begin{itemize}
\item \textbf{binary string}/\textbf{bit string}:
a finite (possibly empty) sequence of elements from $\{0, 1\}$.
The set of all bit strings is denoted by $\{0, 1\}^*$.

Examples: $0$; $1$; $101$; $00101$; $1010$; $1101$; $\eps$ (the empty 
string).
\item \textbf{length of a bit string $x$}: the number of bits
  in $x$, denoted by $|x|$.\\
Examples: $|0| = 1$; $|1| = 1$; $|101| = 3$; $|00101| = 5$. 
\item \textbf{$n_0(x), n_1(x)$}: the number of $0$- and $1$-bits
  in a bit string $x$.

 Examples:
  $n_0(010) = 2$; $n_1(010) = 1$; $n_0(1111) = 0$; $n_1(1111) = 4$.
\item \textbf{$n$-bit string}: a bit string of length $n$, for
  $n \in \N$.  The set of all $n$-bit strings is denoted by
  $\Sigma^n$.

Examples: $001$ is a $3$-bit string; $00101$ is a $5$-bit string;
  $\Sigma^2 = \{00, 01, 10, 11\}$.
\item \textbf{prefix}: a bit string
$x \in \Sigma^*$ is a \emph{prefix} of another binary string 
$y \in \Sigma^*$ if $y$ is of the form $y = xz$, for
some $z \in \Sigma^*$. Here, $xz$ denotes the bit string obtained
by concatenating $x$ with $z$.

Example:
$0$, $01$, $010$, and $0100$ are prefixes of $0100$,
but $1$ and $00$ are not.
\end{itemize}

Next, we explain our coding theoretic vocabulary. In
the following, let $X$ be a finite or countable set.
\begin{itemize}
\item \textbf{code for $X$}/\textbf{encoding of $X$}: an 
injective function $C \from X \to \{0, 1\}^*$ that
assigns a unique finite bit string to every element of $X$.\\
Example: Let $X = \{a,b,c\}$.
The function $C: a \mapsto 01, b \mapsto 1110, c \mapsto 10$ is
a code for $X$, but the function 
$C': a \mapsto 01, b \mapsto 1110, c \mapsto 01$ is not.
\item \textbf{codewords of a code $C\from X \to \{0, 1\}^*$}:
the elements of the range of $C$.

Example: The codewords of
$C: a \mapsto 01, b \mapsto 1110, c \mapsto 10$ are
$01$, $1110$, and $10$.
\item \textbf{partial code $C\from X\nrightarrow \{0,1\}^*$ for $X$}: a 
code that is only a partial function, \emph{i.e.}, not every element
in $X$ is assigned a codeword.
We will use the convention that
$|C(x)|=\infty$ if $x$ is not in the domain of $C$.

Example:
Let $X = \{a,b,c\}$.
Then $C: a \mapsto 01, b \mapsto 1110$ is
a partial code for $X$ in which there is no codeword for $c$.
We have $|C(a)| = 2$, $|C(b)| = 4$, and $|C(c)| = \infty$.
\item \textbf{prefix-free (partial) code}: a 
(partial) code $C$ in which no codeword is the prefix
of another codeword.

Examples: The code $C_1: a \mapsto 10, b\mapsto 111, c \mapsto 01$
is prefix-free. The code $C_2: a \mapsto 10, b\mapsto 111, c \mapsto 1011$
is not prefix-free, because the codeword for $a$ is a prefix
of the codeword for $c$.
\item \textbf{fixed-length code for a finite set $X$}: 
a prefix-free code for $X$ where each codeword has
length $\lceil\log |X|\rceil$.
It is obtained by enumerating
the elements of $X$ in some order $x_0,x_1,\ldots,x_{|X|-1}$ and
assigning to each $x_i$ the binary 
representation of $i$, padded with
leading zeros to obtain $\lceil\log |X|\rceil$ bits.

Example: Let $X = \{a,b,c,d,e\}$. Then, a possible
fixed-length code for $X$ is
$C: a \mapsto 000, b \mapsto 001, c \mapsto 010, d \mapsto 011,
e \mapsto 100$.
\end{itemize}

It is helpful to visualize prefix-free
codes as (rooted ordered) binary trees whose leaves are labelled with
the elements of $X$.  The codeword for a given
$x\in X$ is obtained by tracing the root-to-leaf path leading to $x$
and outputting a 0 each time this path goes from a parent to its left
child, and a 1 each time it goes to a right child. (See
\figref{bintree}.)

\begin{figure}
  \centering{\includegraphics{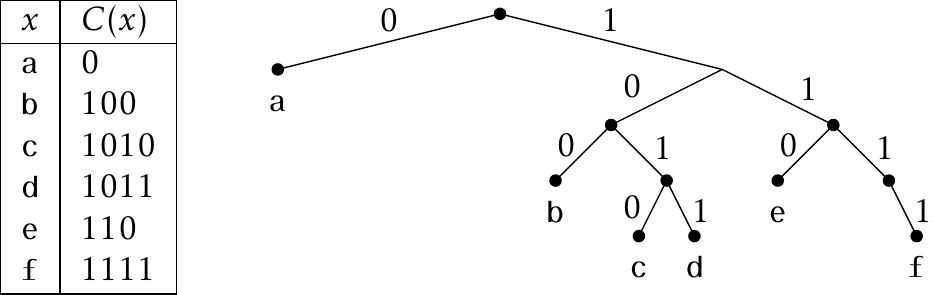}}
  \caption{A prefix-free code for 
    $X=\{\mathtt{a},\mathtt{b},\mathtt{c},\mathtt{d},\mathtt{e},
    \mathtt{f}\}$
    and the corresponding leaf-labelled binary tree (which can also be
    viewed as a partial prefix-free code for 
    $\{\mathtt{a},\mathtt{b},\mathtt{c},\ldots,\mathtt{z}\}$).}
  \figlabel{bintree}
\end{figure}

We claim that if a code $C$ is prefix-free, then for any
$k \in \N$, the code $C$ has 
not more than $2^k$ codewords of length at most $k$.
To see this, we modify $C$ 
into a code $\widehat C$, in which every codeword of length
$\ell <k$ is extended to a word of length exactly $k$ by 
appending $k-\ell$ zeros.
Since $C$ is prefix-free, $\widehat C$
is a valid code, i.e., all codewords of $\widehat C$ 
are pairwise distinct. 
Furthermore, the number of codewords in $C$ with length
at most $k$ equals the number of codewords in $\widehat C$ of 
length exactly $k$.
Since all codewords of length $k$ in $\widehat C$ are pairwise
distinct, there are at most $2^k$ of them.
To illustrate the proof, consider the code
$C: a \mapsto 0, b \mapsto 100, c \mapsto 1010, d \mapsto 1011,
e \mapsto 110, f \mapsto 1111$ from \figref{bintree}.
Our claim says that $C$ has at most $2^3 = 8$ codewords of length
at most $3$. The modified code
$\widehat{C}$ is as follows.
$\widehat{C}: a \mapsto 000, b \mapsto 100, c \mapsto 1010, d \mapsto 1011,
e \mapsto 110, f \mapsto 1111$. We see that $\widehat{C}$ is indeed
a code, and its
codewords of length exactly $3$ are in correspondence with the codewords
of length at most $3$ in $C$, as claimed.

Finally, we need to review some probability theory.

\begin{itemize}
\item \textbf{probability distribution $p$ 
on $X$}: a function $p: X \rightarrow [0,1]$ with
  $\sum_{x \in X} p(x) = 1$. We sometimes write
  $p_x$ instead of $p(x)$. We again emphasize
  that $X$ is a finite or countable set.

Example: for $X = \{a,b,c\}$, the function
  $p: a \mapsto 1/3, b \mapsto 1/2, c \mapsto 1/6$ is 
  a probability distribution, but 
  $p': a \mapsto 1/3, b \mapsto 1/2, c \mapsto 1/5$ is not.
  For $X = \{1, 2, \dots \}$, the function
  $p(x) = 1/2^{x}$ is a probability distribution.
\item \textbf{uniform distribution on finite set $X$}:
  the probability distribution $p: X \rightarrow [0,1]$
  given by $p_x = 1/|X|$, for all $x \in X$.
\item \textbf{Bernoulli distribution on $n$-bit strings
with parameter $\alpha \in [0,1]$}: the probability
distribution $p$ on $X = \{0, 1\}^n$ with
$p_x = (1-\alpha)^{n_0(x)}\alpha^{n_1(x)}$.
In other words, a bit string $x$ is sampled
by setting each bit to $1$ with probability $\alpha$ and to 
$0$ with probability $1-\alpha$,
independently of the other bits.
We write $\mathrm{Bernoulli}(\alpha)$ for $p$.
\end{itemize}
The following lemma provides the foundation on which this
survey is built. The lemma is folklore, but as we will see in
the following sections, it has an incredibly wide range of
applications and can lead to surprisingly powerful results.
\begin{lem}[Uniform Encoding Lemma]\lemlabel{uel}
  Let $X$ be a finite set and 
  $C\from X\nrightarrow \{0,1\}^*$ a partial prefix-free
  code. If an element $x\in X$ is chosen uniformly at random, then
  \[
    \Pr\{|C(x)|\le \log|X|-s\}\le 2^{-s} \enspace .
  \]
\end{lem}

\begin{proof}
  We call a codeword of
  $C$ \emph{short} if it has length at most 
  $k = \lfloor \log|X|-s \rfloor$. Above, we observed that
  $C$ has at most $2^{k}$ short codewords.  Since $C$ is injective,
  each codeword has at exactly one preimage in $X$.  
  Since $x$ is chosen uniformly at
  random from $X$, the probability that it is the preimage of a 
  short codeword is at most
  \[
  \frac{2^k}{|X|} \leq \frac{2^{\log|X|-s}}{|X|} = 2^{-s} \enspace
  . 
  \]
\end{proof}

\subsection{Runs in Binary Strings}

To finish our tutorial, we give a simple use of the 
Uniform Encoding Lemma. Let $x$ be a bit string. 
A \emph{run} in $x$ is a consecutive sequence of
one bits. For example,
$0110111101011111$ contains runs of length
$2$, $4$, $1$, and $5$,  and every substring
of a run is also a run. We now show that
a $n$-bit string that is
chosen uniformly at random is unlikely to contain
a run of length significantly more than $\log n$.  (See
\figref{runs-i}.)

\begin{figure}
  \centering{\includegraphics{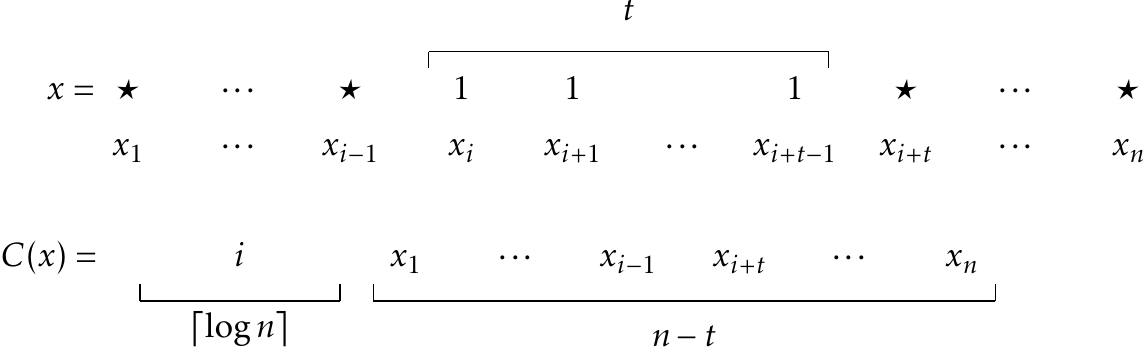}}
  \caption{Illustration of \thmref{runs-i} and its proof.}
  \figlabel{runs-i}
\end{figure}

\begin{thm}\thmlabel{runs-i}
  Let $x=(x_1,\ldots,x_n)\in\{0,1\}^n$ be chosen uniformly at random
  and let $t = \Ceil{\ceil{\log n} + s}$. Then, the probability that
  $x$ contains a run of length $t$
  is at most $2^{-s}$.
\end{thm}

\begin{proof}
  We construct a partial prefix-free
  code for strings with a run of length $t$. For such a string 
  $x=(x_1,\ldots,x_n)$, let $i$ be the minimum index with
  $x_i=x_{i+1}=\cdots=x_{i+t-1}=1$. The codeword $C(x)$ for $x$ is the
  binary string that consists of the ($\ceil{\log n}$-bit binary
  encoding of the) index $i$ followed by the $n-t$ bits
  $x_1,\dots,x_{i-1},x_{i+t},\dots,x_n$. (See \figref{runs-i}.)
  For example, for $n = 8$ and $t = 4$, we encode
  $x = 10111110$ as $010\,10\,10$, since
  the first run of length $4$ in $x$ is at position $2$,
  which is $010$ as a $3 = \lceil \log 8 \rceil$-bit number.

  Observe that $C(x)$ has length 
  \[
    |C(x)| = \ceil{\log n} + n - t \le n-s \enspace .
  \]
  To see that $C$ is injective, 
  we argue that we can obtain $(x_1,\ldots,x_n)$ uniquely 
  from $C(x)$: 
  The first $\ceil{\log n}$
  bits from $C(x)$ tell us, in binary, a position $i$ for
  which $x_i = x_{i + 1} = \dots = x_{i + t - 1} = 1$; the following
  $n - t$ bits in $C(x)$ contain the remaining 
  bits
  $x_1, x_2, \dots, x_{i - 1}, x_{i + t}, x_{i + t + 1}, \dots,
  x_n$. Thus, $C$ is a partial code whose domain are the $n$-bit strings
  with a run of length $t$.  Also,
  $C$ is prefix-free, as all codewords have the same 
  length.
  Recall from our convention that $|C(x)| = \infty$ if
  $x$ does not have a run of length $t$.

  Now, there are $2^n$ $n$-bit strings.
  Therefore, by the Uniform Encoding Lemma, the probability that 
  a uniformly random $n$-bit string has a run of length $t$ 
  is at most
  \[
    \Pr\{|C(x)|\le n-s\} \le 2^{-s} \enspace . 
  \]
\end{proof}

The proof of \thmref{runs-i} contains the main ideas
used in most encoding arguments: 
\begin{enumerate}
\item The arguments usually show that a particular \emph{bad event} is
  unlikely. In \thmref{runs-i} the bad event is the occurrence of a
  substring of $t$ consecutive ones.

\item We use a partial prefix-free code whose domain is the bad
  event, whose elements we call the \emph{bad outcomes}. In this case, the code
  $C$ only encodes strings containing a run of length $t$,
  and a particular such string is a
  bad outcome.

\item The code begins with a concise description of the bad
  outcome, followed by a straightforward encoding of the
  information that is not implied by the bad outcome. In
  \thmref{runs-i}, the bad outcome is completely described by the index
  $i$ at which the run of length $t$ begins, and this implies that
  the bits $x_i,\ldots,x_{i+t-1}$ are all equal to 1, so these bits 
  can be omitted in the second part of the codeword.
\end{enumerate}

\subsection{A Note on Ceilings}\seclabel{ceilings}

Note that \thmref{runs-i} also has an easy proof using the union
bound: If we let $\mathcal{E}_i$ denote the event
$x_i=x_{i+1}=\cdots=x_{i+t}=1$, then
\begin{align*}
  \Pr \left\{\bigcup_{i=0}^{n-t-1} \mathcal{E}_i\right\}
  & \le \sum_{i=0}^{n-t-1} \Pr \{\mathcal{E}_i\} && \text{(using the union bound)}\\
  & = \sum_{i=0}^{n-t-1} 2^{-t}  &&\text{(using the independence of the $x_i$'s)}\\
  & \le n2^{-t} && \text{(the sum has $n-t\le n$ identical terms)}\\
  & \le n2^{-\lceil\log n\rceil-s}&& \text{(using the definition of $t$)}\\
  & \le 2^{-s} \enspace. && \text{($*$)}
\end{align*}
This traditional proof also works with the sometimes smaller value
$t=\ceil{\log n+s}$ (note the lack of a ceiling over $\log n$), 
in which case the final inequality ($*$) holds with an equality.

In the encoding proof of \thmref{runs-i}, the ceiling on the
$\log n$ is an artifact of encoding the integer $i$ which comes
from a set of size $n$. When sketching an encoding argument, we
think of this as requiring $\log n$ bits. However, when 
the time comes to write down a
careful proof we need a ceiling over this term as bits are a discrete
quantity.

In \secref{el}, however, we will 
formally justify that the informal
intuition we use in blackboard proofs is actually valid; we can think
of the encoding of $i$ using $\log n$ bits even if $\log n$ is not an
integer.  In general we can imagine encoding a choice from among $r$
options using $\log r$ bits for any $r\in\N$.  From this point
onwards, we omit ceilings this way in all our theorems and
proofs. This simplifies calculations and provides tighter results.
For now, it allows us to state the following cleaner version of
\thmref{runs-i}:

\begin{customthm}{\ref*{thm:runs-i}b}\thmlabel{runs-ii}
  Let $x\in\{0,1\}^n$ be chosen uniformly at random and let
  $t = \ceil{\log n + s}$. Then, the probability that $x$ contains a
  run of length $t$ is at most $2^{-s}$.
\end{customthm}

\section{More Background and Preliminaries}
\seclabel{background}
\subsection{Encoding Sparse Bit Strings}\seclabel{sparse-bit-strings}

At this point we should also point out an extremely useful trick 
for encoding sparse bit strings. For any $\alpha\in(0,1)$,
there exists a code $C_\alpha\from \{0,1\}^n\rightarrow \{0,1\}^*$
such that, for any bit string $x\in\{0,1\}^n$ having $n_1(x)$ ones and
$n_0(x)$ zeros,
\begin{equation}
  |C_\alpha(x)| = \left\lceil n_1(x)\log \frac{1}{\alpha} + 
    n_0(x)\log\frac{1}{1-\alpha} \right\rceil \enspace.
  \eqlabel{sparse-bitstring-i}
\end{equation}
This code is the \emph{Shannon-Fano code}
for $\mathrm{Bernoulli}(\alpha)$ bit strings of length $n$
\cite{fano:transmission,shannon:mathematical}. 
More generally, for any probability density $p : X \to [0, 1]$,
there is a Shannon-Fano code $C : X \to \{0, 1\}^*$ 
such that
\[
  |C(x)| = \ceil{\log (1/p_x)} \enspace .
\]
Moreover, we can construct such a code deterministically, even when
$X$ is countably infinite.

Again, as explained in \secref{el}, we can omit the ceiling in the
expression for $|C_\alpha(x)|$.  This holds for any value of $n$. In
particular, for $n=1$, it gives a ``code'' for a single
bit where the cost of encoding a 1 is $\log(1/\alpha)$ and the cost of
encoding a 0 is $\log(1/(1-\alpha))$.  Indeed, the ``code'' for bit
strings of length $n>1$ is what we get when we apply this 1-bit
code to each bit of the bit string.

If we wish to encode bit strings of length $n$ and we know in advance
that the strings contain exactly $k$ one bits, then we can obtain
an optimal code 
by taking $\alpha=k/n$. The resulting fixed length
code has length
\begin{equation}
  k\log (n/k) + (n-k)\log(n/(n-k))  \eqlabel{sparse-bitstring-ii} \enspace.
\end{equation}
Formula~(\ref{eq:sparse-bitstring-ii}) brings us to our next 
topic: binary entropy.

\subsection{Binary Entropy}

The \emph{binary entropy function} $H\from (0,1)\to(0,1]$ is defined
by
\[
  H(\alpha) = \alpha\log(1/\alpha) + (1-\alpha)\log(1/(1-\alpha)) 
\]
and it will be quite useful.  The binary entropy function and two upper
bounds on it that we derive below are illustrated in \figref{entropy}.
\begin{figure}
  \centering{\includegraphics[scale=0.9]{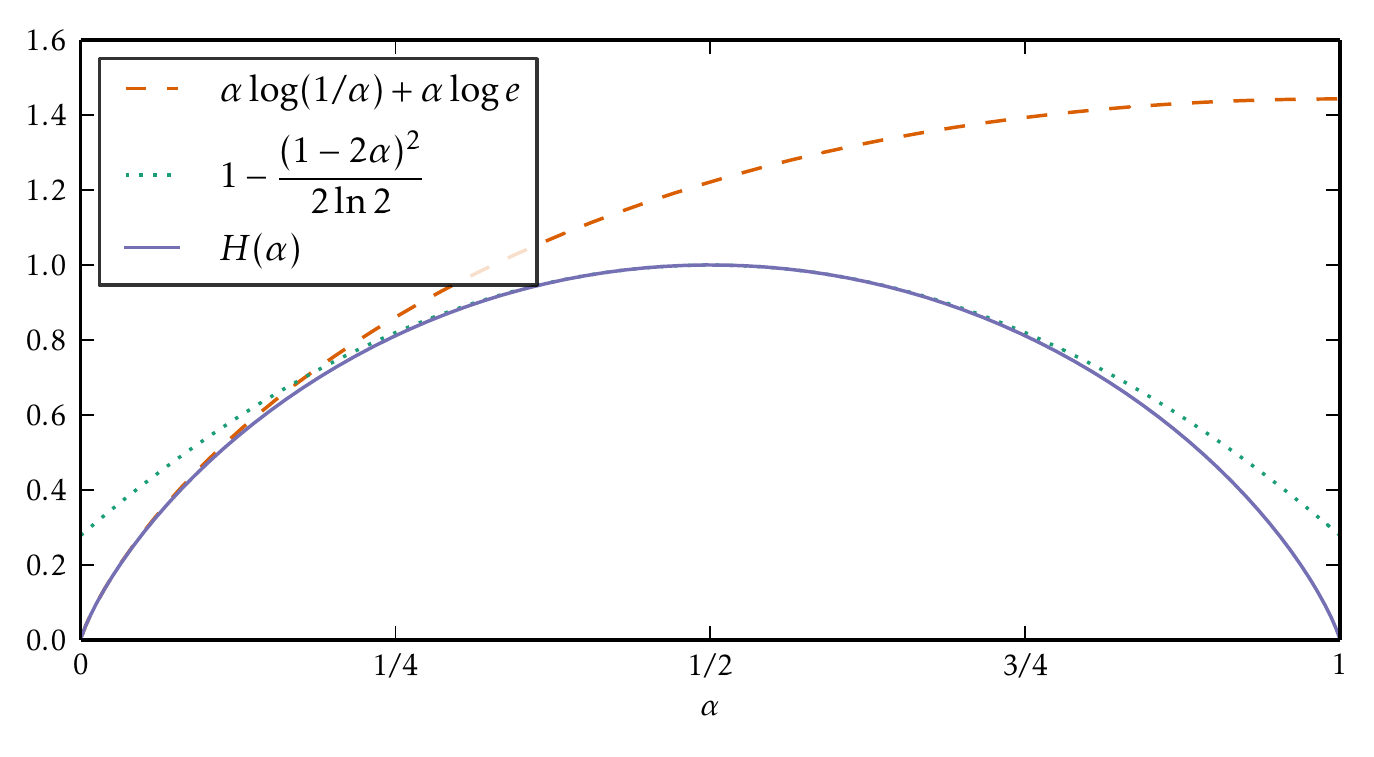}}
  \caption{Binary entropy, $H$, and two useful approximations.}
  \figlabel{entropy}
\end{figure}
We have already encountered a quantity that can be
expressed in terms of the binary entropy.  From
\eqref{sparse-bitstring-ii}, a bit string of length $n$ with
exactly $k$ one bits can be encoded with a fixed-length code of 
$nH(k/n)$ bits.

The binary entropy function can be difficult to work with, so it is
helpful to have some manageable approximations.  One of these is
derived as follows:
\begin{align}
  H(\alpha) & = \alpha\log(1/\alpha) + (1-\alpha)\log(1/(1-\alpha)) 
      = \alpha\log(1/\alpha) + (1-\alpha)\log(1+\alpha/(1-\alpha)) 
      \notag \\
            & \le \alpha\log(1/\alpha) + (1-\alpha)\cdot 
            \alpha/(1-\alpha) \log e 
            \le \alpha\log(1/\alpha) + \alpha\log e 
            \eqlabel{entropy-i} 
\end{align}
since $1+x\le e^x$ for all $x\in\mathbb{R}$. 
Inequality~(\ref{eq:entropy-i}) is a
useful approximation when $\alpha$ is close to zero, in which case 
$H(\alpha)$ is also close to zero.

For $\alpha$ close to $1/2$ (in which case $H(\alpha)$ is close to 1),
we obtain a good approximation from the Taylor series expansion at
$1/2$.  Indeed, a simple calculation shows that
\[
  H'(\alpha) = \log (1/\alpha) - \log (1/(1-\alpha))
\]
and that
\[
  H^{(i)}(\alpha) = \frac{(i-2)!}{\ln 2} \left( \frac{(-1)^{i-1}}{\alpha^{i-1}} - \frac{1}{(1-\alpha)^{i-1}}\right) ,
\]
for $i \geq 2$. Hence, $H^{(i)}(1/2) = 0$, for $i \geq 1$ odd, and
\[
  H^{(i)}(1/2) = -\frac{(i-2)!\, 2^i}{\ln 2}, 
\]
for $i \geq 2$ even. The Taylor series expansion at $1/2$ now gives
\begin{align}
  H(1/2 + \beta) &= 
  H(1/2) + \sum_{i = 1}^{\infty} \frac{H^{(i)}(1/2)}{i!}
     \beta^i \notag
 = 
1 - \frac{1}{\ln 2}\sum_{i = 1}^{\infty} \frac{(2i-2)! 2^{2i}}
{(2i)!}
     \beta^{2i}
= 1-\frac{1}{2\ln 2}\sum_{i=1}^{\infty}\frac{(2\beta)^{2i}}{i(2i-1)} .
              \notag \\ 
\intertext{ In particular, for $\beta=\eps/2$,}
H\left(\frac{1+\eps}{2}\right) & = 1-\frac{1}{2\ln 2}
\sum_{i=1}^{\infty}\frac{\eps^{2i}}{i(2i-1)} 
            < 1-\frac{\eps^2}{2\ln 2} \enspace .\eqlabel{entropy-ii}
\end{align}

\subsection{Basic Chernoff Bound}

With all the pieces in place, we can now give an encoding argument for
a well-known and extremely useful result typically attributed to
Chernoff~\cite{chernoff:bound}.

\begin{thm}\thmlabel{chernoff-basic}
  Let $x\in\{0,1\}^n$ be chosen uniformly at random. Then, for any
  $\eps \geq 0$,
  \[
    \Pr\left\{n_1(x) \leq (1-\eps)\frac{n}{2}\right\} 
    \le e^{-\eps^2n/2} \enspace .
  \]
\end{thm}

\begin{proof}
  Encode the bit string $x$ using a Shannon-Fano code $C_\alpha$ with
  $\alpha=(1-\eps)/2$. 
  Then, the length of the codeword for $x$ is
  \[
    |C_\alpha(x)| = n_1(x)\log(1/\alpha) + n_0(x)\log (1/(1-\alpha))
    \enspace .
  \]
  Since $\alpha < 1/2$, we have $\log(1/\alpha) > \log(1/(1-\alpha))$, 
  so for fixed $n$, the codeword length $|C_\alpha(x)|$ is 
  increasing in $n_1(x)$ and
  becomes maximal when $n_1(x)$ is maximal. 
    Thus, if $n_1(x) \le \alpha n$,
  then
  \begin{align*}
    |C_\alpha(x)|  \le \alpha n\log(1/\alpha) + 
    (1-\alpha)n\log(1/(1-\alpha))
                   = n H(\alpha) \le n\left(1-\frac{\eps^2}{2\ln 2}\right) = n - s \enspace ,
  \end{align*}
  where the second inequality is an application of \eqref{entropy-ii},
  and where $s = \eps^2 n/(2 \ln 2)$.  Now, $x$ was chosen uniformly at
  random from a set of size $2^n$. By the Uniform Encoding Lemma, we
  obtain that
  \[
    \Pr\{n_1(x) \leq \alpha n\} \le \Pr\{|C_\alpha(x)| \le n - s\} \le 2^{-s} = e^{-\eps^2n/2} \enspace . 
  \]
\end{proof}
In \secref{chernoff}, after developing a Non-Uniform Encoding Lemma,
we will extend this argument to $\mathrm{Bernoulli}(\alpha)$ bit
strings. 

\subsection{Factorials and Binomial Coefficients}
\seclabel{stirling}

Before moving on to some more advanced encoding arguments, it will be
helpful to remind the reader of a few inequalities that can be derived
from Stirling's Approximation of the factorial~\cite{robbins:stirling}.  
Recall that Stirling's Approximation states that
\begin{equation}
  n! = \left(\frac{n}{e}\right)^n\sqrt{2\pi n}\left(1+\varTheta\left(\frac{1}{n}\right)\right) \enspace .
  \eqlabel{stirling}
\end{equation}

In many cases, we are interested in representing a set of size $n!$
using a fixed-length code.  By \eqref{stirling}, and using once again
that $1+x \leq e^x$ for all $x \in \R$ as well as
$1+x \geq e^{x/2}$ for $x \in [0, 1/2]$, the length 
of the codewords in such a code is
\begin{align}
  \log n!
  & = n\log n - n\log e + (1/2)\log n + \log \sqrt{2 \pi} + \log(1+\varTheta(1/n)) \notag \\
  & = n\log n - n\log e + (1/2)\log n + \log \sqrt{2 \pi} + \varTheta(1/n) \notag \\
  & = n\log n - n\log e + \varTheta(\log n) \eqlabel{stirling-loose}
    \enspace .
\end{align}

We are sometimes interested in codes for the $\binom{n}{k}$ subsets of
$k$ elements from a set of size $n$. Note that there is an easy
bijection between such subsets and binary strings of length $n$ with
exactly $k$ ones. Therefore, we can represent these using the
Shannon-Fano code $C_{k/n}$ and each of our codewords will have length
$nH(k/n)$.  In particular, this implies that
\begin{equation}
  \log\binom{n}{k} \le nH(k/n) \le k\log n - k\log k + k\log e 
  \eqlabel{log-n-choose-k}
\end{equation}
where the last inequality is an application of \eqref{entropy-i}. The
astute reader will notice that we just used an encoding argument to
prove an upper-bound on $\binom{n}{k}$ \emph{without using the 
formula $\binom{n}{k}=\frac{n!}{k! (n - k)!}$}. Alternatively, we 
could obtain a slightly worse bound by applying 
\eqref{stirling} to this formula.

\subsection{Encoding the Natural Numbers}

So far, we have only explicitly been concerned with codes for finite
sets. In this section, we give an outline of some
prefix-free codes for the set of natural numbers. Of course, if $p :
\N \to [0, 1]$ is a probability density, then the Shannon-Fano code
for $p$ could serve. However, it seems easier to simply design our
codes by hand, rather than find appropriate distributions.

A code is prefix-free if and only if any message consisting of a
sequence of its codewords can be decoded unambiguously and
instantaneously as it is read from left to right: Consider some
sequence of codewords $M = y_1 y_2 \cdots y_k$ from a prefix-free code
$C$. Since $C$ is prefix-free, then $C$ has no codeword $z$ which is a
prefix of $y_1$, so reading $M$ from left to right, the first codeword
of $C$ which we recognize is precisely $y_1$. Continuing in this
manner, we can decode the whole message $M$. Conversely, if for each
codeword $y$ of $C$, a message consisting of this single codeword can
be decoded unambiguously and instantaneously from left to right, we
know that $y$ has no prefix among the codewords of $C$, \emph{i.e.}
$C$ is prefix-free. 
This idea allows us to more easily design
the codes in this section, which were originally given by
Elias~\cite{elias:coding}.

The \emph{unary encoding} of an integer $i \in \N$, denoted by $U(i)$,
begins with $i$ 1 bits which are followed by a 0 bit. This code is
not particularly useful in itself, but it can be improved as follows:
The \emph{Elias $\gamma$-code} for $i$, denoted by $E_\gamma(i)$,
begins with the unary encoding of the number of bits in $i$, and then
the binary encoding of $i$ itself (minus its leading bit). The
\emph{Elias $\delta$-code} for $i$, denoted by $E_\delta(i)$ begins
with an Elias $\gamma$-code for the number of bits in $i$, and then
the binary encoding of $i$ itself (minus its leading bit). This
process can be continued recursively to obtain the \emph{Elias
  $\omega$-code}, which we denote by $E_\omega$. Each of these codes
has a decoding procedure as in the preceding paragraph, which
establishes their prefix-freeness.

The most important properties of these codes are their codeword
lengths:
\begin{align*}
  |U(i)| &= i + 1 \enspace , \\
  |E_\gamma(i)| &= 2 \log i + O(1) \enspace , \\
  |E_\delta(i)| &= \log i + 2 \log \log i + O(1) \enspace , \\
  |E_\omega(i)| &= \log i + \log \log i + \dots + \underbrace{\log \cdots \log}_{\text{$\log^* i$ times}}i + O(\log^* i) \enspace .
\end{align*}
Here, $\log^* i$ denotes the number of times we need to apply
the function $x \mapsto \log x$ to the integer $i$ until we obtain
a number that is smaller than $2$.
It may be worth noting that the lengths of unary codes correspond to
the lengths of Shannon-Fano codes for a geometric distribution with
density $p_i = 1/2^{i + 1}$, that is, 
\[
  \log (1/p_i) = \log 2^{i + 1} = |U(i)| \enspace ,
\]
and the lengths of Elias $\gamma$-codes correspond to the lengths of
Shannon-Fano codes for a discrete Cauchy distribution with density
$p_i = c/i^2$ for a normalization constant $c$, that is,
\[
  \log (1/p_i) = 2 \log i - \log c = |E_\gamma(i)| + O(1) \enspace .
\]
The lengths of Elias $\delta$-codes and $\omega$-codes do not seem to
arise as the lengths of Shannon-Fano codes for any 
named distributions.

\section{Applications of the Uniform Encoding Lemma}
\seclabel{applications-i}

We now start with some applications of the Uniform Encoding Lemma. In
each case, we will design and analyze a partial prefix-free code
$C\from X\nrightarrow \{0, 1\}^*$, where $X$ depends on
the context.

\subsection{Graphs with no Large Clique or Independent Set}\seclabel{ramsey}

The \emph{Erd\H{o}s-R\'enyi random graph} $G_{n,p}$ is the
probability space on
graphs with vertex set $V=\{1,\ldots,n\}$ in which each edge $\{u,
w\} \in \binom{V}{2}$ is present with probability $p$ and absent with
probability $1-p$, independently of the other edges.  
Erd\H{o}s~\cite{erdos:some} used the 
random graph $G_{n,\frac{1}{2}}$ to prove
that there are graphs that have neither a large clique nor a large 
independent set. Here we show how this can be done using an encoding
argument.

\begin{thm}\thmlabel{erdos-renyi-i}
  For $n \ge 3$ and $s \in \N$, the probability that $G \in G_{n,\frac{1}{2}}$
  contains a clique or an independent set of size $t = \ceil{3\log n +
    \sqrt{2s}}$ is at most $2^{-s}$.
\end{thm}

\begin{proof}
  This encoding argument compresses the $\binom{n}{2}$ bits
  of $G$'s adjacency matrix, as they appear in row-major order.
  
  Suppose the graph $G$ contains a clique or an independent set $S$ of size
  $t$. The encoding $C(G)$ begins with a bit indicating whether $S$
  is a clique or independent set; followed by the set of vertices of $S$; 
  then the
  adjacency matrix of $G$ in row major-order, omitting
  the $\binom{t}{2}$ bits implied by the edges or non-edges in
  $S$. Such a codeword has length
  \begin{equation} 
    |C(G)|  = 1 + t\log n + \binom{n}{2}-\binom{t}{2} \enspace .  
    \eqlabel{tst-1}
  \end{equation}
  Before diving into the detailed arithmetic, we intuitively
  argue why we're heading in the right direction: Roughly,
  \eqref{tst-1} is of the form:
  \[ |C(G)|  = \binom{n}{2} + t\log n - \varOmega(t^2) \enspace . \]
  That is, we need to invest $O(t\log n)$ bits
  to encode the vertex set of a clique or an independent set of 
  size $t$, but we save $\varOmega(t^2)$ bits in the encoding of $G$'s 
  adjacency matrix.
  Clearly, for $t>c\log n$, with $c$ sufficiently large, this has the form 
  \[ |C(G)|  = \binom{n}{2} - \varOmega(t^2) \enspace . \]
  At this point, it is just a matter of pinning down the dependence on $c$.
  A detailed calculation beginning from \eqref{tst-1} gives
  \begin{align*}
    |C(G)| 
     = \binom{n}{2} + 1 + t\log n - (1/2)(t^2 - t) 
     = \binom{n}{2} + 1 - (1/2)(t^2 - t - 2t \log n) \enspace .
  \end{align*}
  The function $f(x) = (1/2)(x^2 - x - 2x \log n) - 1$ is increasing
  for $x \geq \log n + 1/2$, so recalling that
$t = \ceil{3\log n + \sqrt{2s}}$, we get
  \begin{align*}
    f(t) &\ge f(3\log n + \sqrt{2s}) \\
    &= (1/2)(9 \log^2 n + 6 \sqrt{2s} \log n + 2s - 3 \log n - \sqrt{2s} - 6 \log^2 n - 2 \sqrt{2s} \log n) - 1 \\
    &= (1/2)(3 \log^2 n + 4 \sqrt{2s} \log n - 3 \log n - \sqrt{2s}) + 
    s - 1  \ge s
  \end{align*}
  for $n \ge 3$. Therefore, our code has length
  \[
    |C(G)| = \binom{n}{2} - f(t) \le \binom{n}{2} - s \enspace .
  \]
  Applying the Uniform Encoding Lemma completes the proof.
\end{proof}

\begin{rem}
  The bound in \thmref{erdos-renyi-i} can be strengthened a little,
  since the elements of $S$ can be encoded using only
  $\log\binom{n}{t}$ bits, rather than $t\log n$.  With a more careful
  calculation, using \eqref{log-n-choose-k}, the proof then works with
  $t = 2\log n +O(\log\log n) + \sqrt{s}$. This comes closer to 
  Erd\H{o}s's original result, which was at the threshold $2\log n - 2\log\log n +
  O(1)$~\cite{erdos:some}.
\end{rem}

\subsection{Balls in Urns}
\seclabel{urns}

The random experiment of  throwing $n$ balls uniformly and
independently at random into $n$ urns is a useful abstraction of many
questions encountered in algorithm design, data structures, and 
load-balancing~\cite{mitzenmacher.upfal:probability,motwani.raghavan:randomized}.  Here we
show how an encoding argument can be used to prove the classic result
that, when we do this, every urn contains $O(\log n/\log\log
n)$ balls, with high probability.

\begin{thm}\thmlabel{urns}
  Let $n, s \in \N$, and let $t$ be such that $t\log(t/e) \ge \log n + s$.
  Suppose we
  throw $n$ balls independently and uniformly at random into $n$
  urns. Then, for sufficiently large $n$, the probability that any urn
  contains more than $t$ balls is at most $2^{-s}$.
\end{thm}

Before proving \thmref{urns}, we note that, for any constant $\eps >0$
and all sufficiently large $n$, taking
\[
  t = \Ceil{\frac{(1+\eps)\log n}{\log\log n}}
\] 
satisfies the requirements of \thmref{urns},  since then
\begin{align*}
  t \log t &\ge \frac{(1 + \eps) \log n}{\log \log n} 
  \log \left( \frac{(1 + \eps) \log n}{\log \log n}\right) 
  = \frac{(1 + \eps) \log n}{\log \log n} 
  \left( \log\log n - \log\left(\frac{\log\log n}{1 + \eps}\right)
  \right) \\
  &= (1 + \eps) \log n 
  \left( 1 - \frac{\log\left(\frac{\log\log n}{1 + \eps}\right)}{\log\log n}
     \right) 
  = (1+\eps)\log n - o(\log n) \enspace.
\end{align*}
Then,
\[
  t \log (t/e) = t \log t - t \log e = (1+\eps)\log n 
  - o(\log n) \ge \log n + s \enspace ,
\]
for sufficiently large $n$, as claimed.

\begin{proof}[Proof of \thmref{urns}]
  For each $i\in\{1,\ldots,n\}$, let $b_i$ denote the index of the urn
  chosen for the $i$-th ball. The sequence $b = (b_1,\ldots,b_n)$ is
  sampled uniformly at random from a set of size $n^n$, and this
  choice will be used in our encoding argument.

  Suppose that urn $j$ contains $t$ or more balls. Then,
  we encode the sequence $b$
  with the value $j$ ($\log n$ bits), followed by a code that describes $t$ of
  the balls in urn $j$ ($\log\binom{n}{t}$ bits), followed by the remaining $n-t$ values in
  $b$ that cannot be deduced from the preceding
  information ($(n-t)\log n$ bits).  Thus, we get
  \begin{align*}
    |C(b)| &= \log n + \log\binom{n}{t} + (n-t)\log n \\
           &\le \log n + t\log n - t\log t + t\log e + (n-t)\log n
             \tag{using \eqref{log-n-choose-k}} \\
           & = n \log n + \log n - t\log t + t\log e \\
           & \le n \log n - s \tag{by the choice of $t$} \\
           & = \log n^n - s
  \end{align*}
  bits. We conclude the proof by applying the Uniform Encoding Lemma.
\end{proof}

\subsection{Linear Probing}

Studying balls in urns as in the previous section is useful when 
analyzing hashing with chaining
(see \emph{e.g.}~Morin~\cite[Section~5.1]{morin:open}). A more practically
efficient form of hashing is \emph{linear probing}.  In a linear
probing hash table, we hash the elements of the set
$X = \{x_1, \ldots, x_n\}$ into a hash table of size $m=cn$, for some
fixed $c> 1$. We are given a hash function
$h : X \to \{1, \ldots, m\}$ which we assume to be a uniform random
variable. To insert $x_i$, we try to place it at table position
$j=h(x_i)$. If this position is already occupied by one of
$x_1,\ldots,x_{i-1}$, we try table location
$(j+1)\bmod m$, followed by $(j+2)\bmod m$, and so on, until 
we find an empty spot
for $x_i$.  
To find a given element $x \in X$ in the hash table, we compute
$j = h(x)$, and we start a linear search from position $j$ 
until we encounter either $x$ or an empty position.
Assuming that the hash table has been created by inserting
the elements from $X$ successively according to the
algorithm above, 
we want to study the expected search time for some item
$x \in X$.

We call a maximal consecutive sequence of occupied table locations a
\emph{block}. (The table locations $m-1$ and $0$ are considered
consecutive.)

\begin{thm}\thmlabel{linear-probing}
  Let $n \in \N$, $c > e$.
  Suppose that a set $X = \{x_1, \dots, x_n\}$ of $n$ items 
  has been inserted into a hash table of size $m = cn$, using linear 
  probing.
  Let $t \in \N$, $t \geq 2$, such that
  \[
    (t-1) \log (c/e) - \log t - 3 \ge s \enspace ,
  \]
  and fix some $x\in X$. 
  Then the probability that the block containing
  $x$ has size exactly $t$
  is at most $2^{-s}$.
\end{thm}

\begin{proof}
  This is an encoding argument for 
  the sequence
  \[
    h = (h(x_1),h(x_2),\ldots,h(x_n)) \enspace ,
  \]
  that is drawn uniformly at random from a set of size
  $m^n = (cn)^n$.
  
  Suppose that $x$ lies in a block with $t$ elements.
  We encode $h$ by the first index $b$ of the block containing $x$;
  followed by the $t-1$ elements $y_1,\dots,y_{t-1}$ of this block
  (excluding $x$); followed by the hash
  values of $x$ and of $y_1,\dots,y_{t-1}$; followed by
  the $n-t$ hash values for the remaining elements in $X$.

  Since the values $h(x),h(y_1),h(y_2),\ldots,h(y_{t-1})$ are in the range
  $b,\ldots,b+t-1$ (modulo $m$), they can be encoded using
  $t\log t$ bits.  Therefore, 
  we obtain a codeword of length 
  \begin{align*}
    |C(h)| & = \overbrace{\log m}^{b} + \overbrace{\log\binom{n}{t-1}}^{y_1,\ldots,y_{t-1}} + \overbrace{t\log t}^{h(x),h(y_1),\ldots,h(y_{t-1})} + \overbrace{(n-t)\log m}^{\text{everything else}} \\
           & \le \log m + (t-1)\log n - 
             (t-1)\log(t-1) + (t-1)\log e + t\log t + (n-t)\log m \tag{by \eqref{log-n-choose-k}}\\
	  &=
         (n-t+1)\log m + (t-1)\log (m/c) + (t-1)\log e + \log (t-1)+ t\log(t/(t-1)) \tag{$m = cn$}\\
           & \leq n\log m - (t-1)\log c + (t-1)\log e + \log t + 3 \\
           & = \log m^n - (t-1)\log(c/e) + \log t + 3 
            \le \log m^n - s \enspace ,
  \end{align*}
  since we assumed that $t$ satisfies
  \[
    (t-1) \log (c/e) - \log t - 3 \ge s 
  \]
  and since for $t \geq 2$, we have
  \[
   t\log\frac{t}{t-1} = \frac{t}{\ln 2} \ln 
   \left(1 + \frac{1}{t-1}\right) 
   \leq \frac{t}{\ln 2} \cdot \frac{1}{t-1}
   \leq 3.
  \]
  The proof is completed by applying the Uniform Encoding Lemma.
\end{proof}

\begin{rem}
  The proof of \thmref{linear-probing} only works if the factor
  $c$ in the size $m = cn$ of the linear probing hash table is 
  $c > e$.
  We know from
  previous analysis that this is not necessary, and that any $c>1$ is
  sufficient~\cite[Theorem~9.8]{flajolet.sedgewick:aofa}. 
  We leave it as an open problem to find an encoding proof
  of \thmref{linear-probing} that works for any $c>1$.
\end{rem}

\begin{cor}
  Let $n \in \N$, $c > e$.
  Suppose that a set $X = \{x_1, \dots, x_n\}$ of $n$ items 
  has been inserted into a hash table of size $m = cn$, using linear 
  probing.
  Fix some $x\in X$. 
  Then, the expected search time for $x$ in the hash table is
  $O(1)$.
\end{cor}
\begin{proof}
  Let $T$ denote the size of the block containing $x$ in the
  hash table. Let $t_0$ be a large enough constant. 
  Then, by \thmref{linear-probing}, the probability
  that $T = t + t_0$ is at most $2^{-t\log(c/e)/2}$,
  since then
  \[(t  + t_0 - 1)\log(c/e) - \log(t+t_0) - 3
  \geq
(t  + t_0)\log(c/e)/2 - 3
    \geq t\log(c/e)/2  \enspace. 
    \]
  Thus, the expected search time for $x$ is
  \begin{align*}
   \E\{T\} &= \sum_{t = 1}^\infty t\Pr\{T = t\} = \sum_{t=1}^{t_0} t\Pr\{T = t\} + \sum_{t = 1}^\infty (t+t_0)\Pr\{T = t + t_0\} \\
            &\le t_0 + \sum_{t = 1}^\infty (t+t_0)2^{-t \log (c/e)/2} 
            = t_0 + \sum_{t = 1}^\infty (t+t_0)(c/e)^{-t/2} = O(1) \enspace . 
  \end{align*}
\end{proof}

\subsection{Cuckoo Hashing}

Cuckoo hashing is relatively new hashing scheme that offers an
alternative to classic perfect hashing \cite{pagh.rodler:cuckoo}. 
We present a clever proof, due to 
P\u{a}tra\c{s}cu~\cite{patrascu:cuckoo}, that cuckoo hashing
succeeds with probability $1-O(1/n)$ 
(see also Haimberger~\cite{haimberger:cuckoo} for 
a more detailed exposition of the argument).

We again hash the elements of the set $X = \{x_1, \ldots, x_n\}$. The
hash table consists of two arrays $A$ and $B$, each of size $m = 2n$, and
two hash functions $h, g : X \to \{1, \ldots, m\}$ which are uniform
random variables. To insert an element $x$ into the hash table, we
insert it into $A[h(x)]$; if $A[h(x)]$ already contains an element
$y$, we insert $y$ into $B[g(y)]$; if $B[g(y)]$ already contains some
element $z$, we insert $z$ into $A[h(z)]$, etc. If an empty location
is eventually found, the algorithm terminates successfully. If the
algorithm runs for too long without successfully completing the
insertion, then we say that the insertion failed, and the hash table
is rebuilt using different newly sampled hash functions. Any element
$x$ either is held in $A[h(x)]$ or $B[g(x)]$, so we can search for $x$
in constant time.  The following
pseudocode describes this procedure more precisely:

\noindent{$\textsc{Insert}(x):$}
\begin{algorithmic}[1]
  \IF{$x = A[h(x)]$ \OR $x = B[g(x)]$}
    \RETURN
  \ENDIF

  \FOR{MaxLoop iterations}
    \IF{$A[h(x)]$ is empty}
      \STATE{$A[h(x)] \gets x$} 
      \RETURN
    \ENDIF
    \STATE{$x \leftrightarrow A[h(x)]$}
    \IF{$B[g(x)]$ is empty}
      \STATE{$B[g(x)] \gets x$}
      \RETURN
    \ENDIF
    \STATE{$x \leftrightarrow B[g(x)]$}
  \ENDFOR
  \STATE{$\textsc{Rehash}()$}
  \STATE{$\textsc{Insert}(x)$}
\end{algorithmic}

The threshold `MaxLoop' is to be specified later. To study the
performance of insertion in cuckoo hashing, we consider the random
bipartite \emph{cuckoo graph} $G = (A, B, E)$, where $|A| = |B| = m$
and $|E| = n$, with each vertex corresponding either to a location in
the array $A$ or $B$ above, and with edge multiset
$E = \{(h(x_i), g(x_i)) : 1 \leq i \leq n\}$.

An \emph{edge-simple} path in $G$ is a path that 
uses each edge at most once. 
One can check that if a successful insertion 
takes at least $2t$ steps, then the cuckoo
graph contains an edge simple path with at least $t$ edges.
Thus, in bounding the length of edge-simple
paths in the cuckoo graph, we bound the worst case insertion time.

\begin{lem}\lemlabel{cuckoo-path-length}
  Let $s \in \N$. Suppose that we insert a set $X = \{x_1, \dots, x_n\}$
  into a hash table using cuckoo hashing.
  Let $G$ be the resulting cuckoo graph. Then, 
  $G$ has an edge-simple path of length at least
  $s + \log n + O(1)$ with probability at most $2^{-s}$.
\end{lem}
\begin{proof}
  We encode $G$ by presenting its set of edges. Since 
  each endpoint of
  an edge is chosen independently and uniformly at random from a set
  of size $m$, the set of all edges is chosen uniformly at random from
  a set of size $m^{2n}$.

  Suppose some vertex $v \in A \cup B$ is the endpoint of an 
  edge-simple path
  of length $t$; such a path has $t + 1$ vertices and $t$ edges. 
  Each edge in the path corresponds to an element in $X$.
  In
  the encoding, we present the indices of the elements in $X$ corresponding
  to the $t$ edges of the path 
  in order; then,
  we indicate whether $v \in A$ or $v \in B$; and we give the $t + 1$
  vertices in order starting from $v$; followed by the remaining
  $2n - 2t$ endpoints of edges of the graph. This code has length
  \begin{align*}
    |C(G)| &= t \log n + 1 + (t + 1) \log m + (2n - 2t) \log m\\
           &= 2n \log m + t \log n - t \log m + \log m + O(1) \\
           &= \log m^{2n} - t + \log n + O(1) \tag{since $m = 2n$} \\
           &\leq \log m^{2n} - s
  \end{align*}
  for $t \geq s + \log n + O(1)$. We finish by applying the Uniform
  Encoding Lemma.
\end{proof}
This immediately implies that a successful insertion takes time at
most $4 \log n + O(1)$ with probability $1 - O(1/n)$. Moreover,
selecting `MaxLoop' to be $4\log n + O(1)$, we see that a rehash
happens only with probability $O(1/n)$.

One can prove that
the cuckoo hashing insertion algorithm fails if and only if  
some subgraph of
the cuckoo graph contains more edges than vertices, since edges
correspond to keys, and vertices correspond to array locations.
\begin{lem}\lemlabel{cuckoo-failure}
  The cuckoo graph has a subgraph with more edges than vertices with
  probability $O(1/n)$. That is, cuckoo insertion
  succeeds with probability $1 - O(1/n)$.
\end{lem}
\begin{proof}
  Suppose that some vertex $v$ is part of a subgraph with more edges
  than vertices, and in particular a minimal such subgraph with $t +
  1$ edges and $t$ vertices. Such a subgraph appears exactly as in
  \figref{cuckoo-cycles}. By inspection, we see that for every such 
  subgraph, there are two edges
  $e_1$ and $e_2$ whose removal disconnects the graph into two paths
  of length $t_1$ and $t_2$ starting from $v$, where $t_1 + t_2 = t -
  1$.

  \begin{figure}
    \centering
    \begin{subfigure}[b]{0.3\textwidth}
      \includegraphics[scale=0.8]{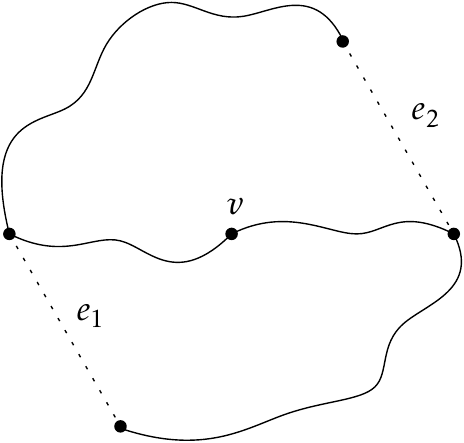}
      \caption{A cycle with a chord.}
    \end{subfigure}
    \quad\quad
    \begin{subfigure}[b]{0.6\textwidth}
      \includegraphics[scale=0.8]{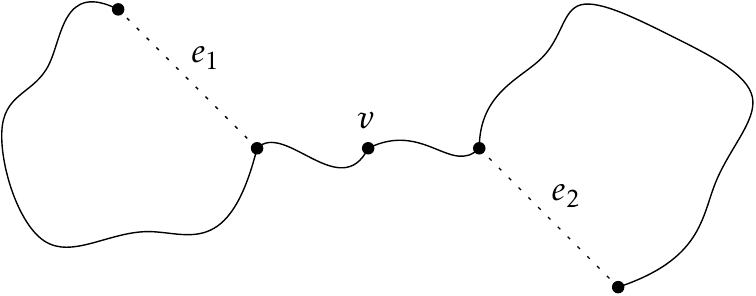}
      \caption{Two cycles connected by a path.}
    \end{subfigure}
    \caption{The potential minimal subgraphs of the cuckoo graph.}
    \figlabel{cuckoo-cycles}
  \end{figure}

  We encode $G$ by giving the vertex $v$ ($\log m$ bits); and presenting Elias
  $\delta$-codes for the values of $t_1$ and $t_2$ and for the
  positions of the endpoints of $e_1$ and $e_2$ ($O(\log t)$ bits); then the indices
  of the edges of
  the above paths in order ($(t-1)\log n$ bits); then the vertices of the paths in 
  order ($(t-1)\log m$ bits);
  and the indices of the edges $e_1$ and $e_2$ ($2 \log n$ bits); and finally the 
  remaining $2n - 2(t + 1)$ endpoints of edges in the graph
  ($(2n - 2(t + 1))\log n$ bits). Such a code has length
  \begin{align*}
    |C(G)| &= \log m + O(\log t) + (t - 1)(\log n + \log m) + 2\log n + (2n - 2(t + 1))\log m \\
           &= 2n \log m + (t + 1) \log n - (t + 2) \log m + O(\log t) \\
           &= 2n \log m + (t + 1) \log n - (t + 2) \log n - t + O(\log t) \tag{since $m = 2n$} \\
           &\le \log m^{2n} - \log n + O(1) \enspace .
  \end{align*}
  We finish by applying the Uniform Encoding Lemma.
\end{proof}

\subsection{2-Choice Hashing}

We showed in \secref{urns} that if $n$ balls are thrown independently
and uniformly at random into $n$ urns, then the maximum number of
balls in any urn is $O(\log n/\log \log n)$ with high probability. In
\emph{2-choice hashing}, each ball is instead given a choice of two
urns, and the urn containing the fewer balls is preferred.

More specifically, we are given two hash functions $h, g : X \to \{1,
\ldots, m\}$ which are uniform random variables. Each value of $h$ and
$g$ points to one of $m$ urns. The element $x \in X$ is added to the
urn containing the fewest elements between $h(x)$ and $g(x)$ during an
insertion. The worst case search time is at most the maximum number of
hashing collisions, or the maximum number of elements in any urn.

Perhaps surprisingly, the simple change of having two choices instead
of only one results in an exponential improvement over the strategy of
\secref{urns}. The concept of 2-choice hashing was first studied by 
Azar~\etal~\cite{azar:multiplechoice}, who showed that the expected
maximum size of an urn is $\log \log n + O(1)$. Our encoding argument
is based on V\"{o}cking's use of witness trees to analyze 2-choice
hashing~\cite{vocking:witness}.

Let $G = (V, E)$ be the random multigraph with $V = \{1, \ldots, m\}$,
where $m = cn$ for some constant $c > 8$, and
$E = \{(h(x), g(x)) : x \in X\}$. Each edge in $E$ is
labeled with the element $x \in X$ that it corresponds to.

\begin{lem}\lemlabel{two-choice-two-cycles}
  The probability that $G$ has a subgraph with more edges than 
  vertices is
  $O(1/n)$.
\end{lem}
\begin{proof}
  The proof is similar to that of \lemref{cuckoo-failure}. More
  specifically, we encode $G$ by giving the same encoding as in
  \lemref{cuckoo-failure}. However, since now $G$ is not bipartite,
  we cannot immediately deduce for an edge $uv$ in the encoding
  which endpoint corresponds to which hash function.
  Thus, for each edge $uv$, we store an additional bit indicating whether $u = h(x)$ and $v = g(x)$, or $u
  = g(x)$ and $v = h(x)$. This needs $t$ additional bits compared to \lemref{cuckoo-failure}. 
Our code thus has length
  \begin{align*}
    |C(G)| &= \log m + O(\log t) + (t - 1)(\log n + \log m) + 2 \log n + (2n - 2(t + 1))\log m + t  \\
           &= 2n \log m - \log n - t \log c + t + O(\log t) 
           \le \log m^{2n} - \log n + O(1) \enspace ,
  \end{align*}
  since $\log c > 1$.
\end{proof}

\begin{lem}\lemlabel{two-choice-component-size}
  $G$ has a component of size at least $(2/\log(c/8))\log n + O(1)$
  with probability $O(1/n)$.
\end{lem}
\begin{proof}
  Suppose $G$ has a connected component with $t$ vertices and at least
  $t-1$ edges. This component has a spanning tree $T$. Pick an
  arbitrary vertex as the root of $T$. To encode $G$, we first specify
  a bit string encoding the shape of $T$.
  This can be done in $2(t - 1)$ bits, tracing a pre-order traversal of $T$, 
  where a 0 bit indicates that the path to the next node goes up, and a 1 bit indicates
  that the path goes down.
  Then, we encode the $t$ vertices of $T$, in the order as they are first encountered by
  the pre-order traveral ($t \log m$ bits). Furthermore, 
  for each edge $e = uv$ of $T$, 
  we store $\log n$ bits encoding the element $x \in X$ 
  corresponding to $e$ and a bit
  indicating whether $u = h(x)$ and $v = g(x)$, or $u
  = g(x)$ and $v = h(x)$. Again, the edges are stored
  in the order and direction as they are encountered by the 
  pre-order traversal. As $T$ has $t-1$ edges, 
  this takes $(t-1)(\log n + 1)$ bits.
  Finally we directly encode the remaining $2(n - t + 1)$ endpoints of edges in $G$,
  in $2(n-t+1)\log m$ bits.
 In total, our code has length
  \begin{align*}
    |C(G)| &= 2(t - 1) + t \log m + (t - 1) (\log n + 1) + 2 (n - t + 1) \log m \\
           &= 2n\log m + t \log c - \log n + 3t - 3 - 2t \log c + 2\log n + 2 \log c \tag{since $m = cn$}\\
           &= \log m^{2n} - t \log c + 3t + \log n + O(1) 
           \le \log m^{2n} - s \enspace ,
  \end{align*}
  as long as $t$ is such that
  \[
    t \geq \frac{s + \log n + O(1)}{\log (c/8)} \enspace .
  \]
  In particular, for $s = \log n$, the Uniform
  Encoding Lemma tells us that $G$ has a component of size at least
  $(2/\log (c/8))\log n + O(1)$ with probability $O(1/n)$.
\end{proof}

Suppose that when $x$ is inserted, it is placed in an urn with $t$
other elements. Then, we say that the \emph{age} of $x$ is $a(x) = t$.
\begin{thm}
  Fix $c > 8$ and
  suppose we insert $n$ elements into a table of size $cn$
  using $2$-choice hashing.
  With probability $1 - O(1/n)$
  all positions in the hash table contain at most 
  $\log \log n + O(1)$ elements.
\end{thm}
\begin{proof}
  Suppose that some element $x$ has $a(x) = t$. This leads to a binary
  \emph{witness tree} $T$ of height $t$ as follows: The root of $T$ is
  the element $x$. When $x$ was inserted into the hash table, it had
  to choose between the urns $h(x)$ and $g(x)$, both of which
  contained at least $t - 1$ elements; in particular, $h(x)$ has a unique
  element $x_h$ with $a(x_h) = t - 1$, and $g(x)$ has a unique 
  element $x_g$
  with $a(x_g) = t - 1$. The elements $x_h$ and $x_g$ become the left and
  the right child of $x$ in $T$. The process continues recursively. If
  some element appears more than once on a level, we only recurse on
  its leftmost occurrence. See \figref{2-choice-example} for an example. 

  \begin{figure}
    \centering
    \includegraphics{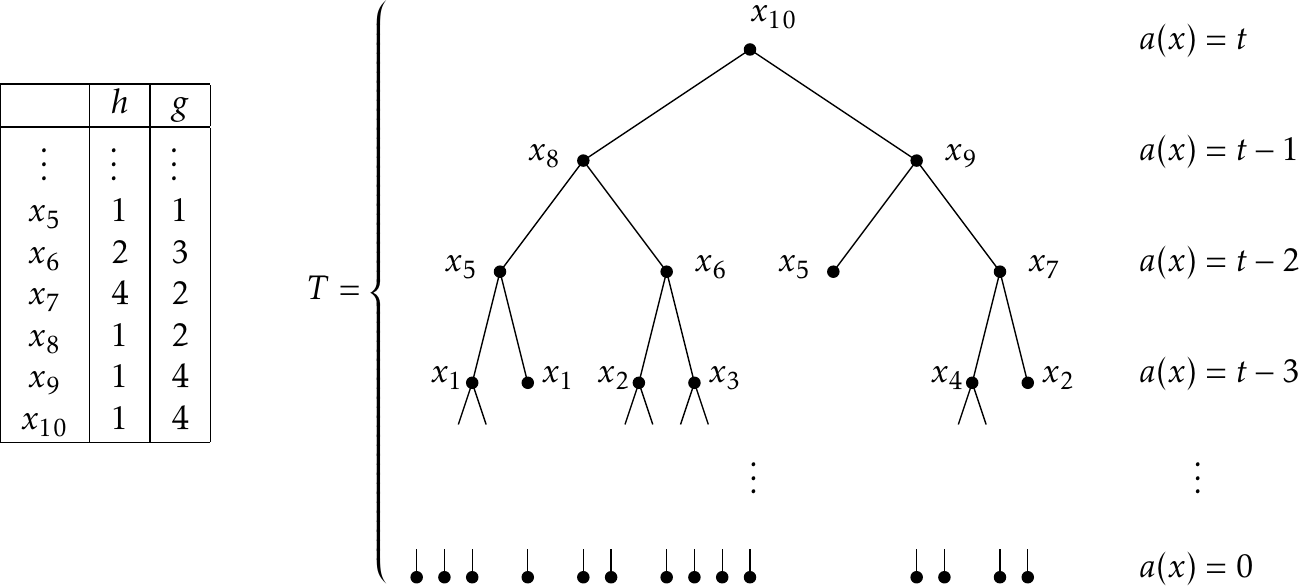}
    \caption{The tree $T$ is a witness tree for the 2-choice hashing
      instance with elements $\dots, x_5, x_6, x_7, x_8, x_9, x_{10}$
      inserted in order and according to the hash functions $h$ and
      $g$.}
    \figlabel{2-choice-example}
  \end{figure}

  Using $T$, we can iteratively define a connected subgraph $G_T$ of $G$.
  Initially, $G_T$ consists of the single node in $V$ corresponding
  to the bucket that contains the root element $x$ of $T$.
  Now, to construct $G_T$, we go through $T$ level by level, 
  starting from the root.
  For $i =  t,\dots, 0$, let
  $L_i$ be all elements in $T$ with age $i$, and let
  $E_i = \{(h(x), g(x)) : x \in L_i\}$ be the corresponding
  edges in $G$.
  When considering $L_i$, we add 
  to $G_T$ all edges in $E_i$, together with their endpoints,
  if they are not in $G_T$ already.
  Since every element appears at most once in $T$, this
  adds $|L_i|$ new edges to $G_T$. The number of vertices in $G_T$
  increases by at most $|L_i|$. In the end, $G_T$ contains
  $\sum_{i = 0}^{t} |L_i|$ edges. Since $G_T$ is connected, 
  with probability $1 - O(1/n)$, the number of edges in $G_T$ does not 
  exceed the number of vertices, by \lemref{two-choice-two-cycles}. 
  We assume that this is the case.
  Since initially $G_T$ had 
  one vertex and zero edges, the iterative procedure must add
  at least $\sum_{i = 0}^{t} |L_i| - 1$ new vertices to $G_T$.
  This means that all nodes in $T$ but one must
  have two children, so we can conclude that $T$ is a
  complete binary tree with at most one subtree removed. It follows
  that $T$ (and hence $G_T$) has at least $2^{t}$ vertices.
  If we choose $t = \ceil{\log \log n + d}$, then 
  $2^t \geq 2^d \log n$. 
  We know from
  \lemref{two-choice-component-size} that this 
  happens with probability $O(1/n)$
  for a sufficiently large choice of the constant $d$.
\end{proof}

\begin{rem}
  The arguments in this section can be refined by more carefully
  encoding the shape of trees using \emph{Cayley's formula}, which says 
  that
  there are $t^{t - 2}$ unrooted labelled trees on $t$
  nodes~\cite{cayley:theorem}. In particular, an unrooted tree with $t$
  nodes and $m$ choices for distinct node labels can be encoded using
  $\log \binom{m}{t} + (t - 2) \log t$ bits instead of
  $t \log m + 2(t - 1)$ bits. We would then recover the same results
  for hash tables of size $m = cn$ with $c > 2e$ instead of $c > 8$. In
  fact, it is known that for \emph{any}
  $c > 0$ searching in 2-choice hashing takes time 
  $1/c + O(\log \log n)$~\cite{berenbrink:densehashing}. We leave it as an
  open problem to find an encoding argument for this result when
  $c > 0$.
\end{rem}

\begin{rem}
  \emph{Robin Hood hashing} is another hashing solution which achieves
  $O(\log \log n)$ worst case running time for all
  operations~\cite{devroye:robin}. The original analysis is difficult,
  but might be amenable to a similar approach as we used in
  this section. Indeed, when a Robin Hood hashing operation takes a
  significant amount of time, a large witness tree is again implied,
  which suggests an easy encoding argument. Unfortunately, this
  approach appears to involve unwieldy hypergraph encoding.
\end{rem}

\subsection{Bipartite Expanders}

Expanders are families of sparse graphs which share many connectivity
properties with the complete graph. These
graphs have received much research attention, and have led to many
applications in computer science. See, for instance, the survey by
Hoory, Linial, and Wigderson~\cite{hoory.linial.ea:expander}.

The existence of expanders was originally
established through probabilistic arguments~\cite{pinsker:expanders}.
We offer an encoding argument to prove that a certain random bipartite
graph is an expander with high probability. There are many different
notions of expansion. We will consider what is commonly known as
\emph{vertex expansion} in bipartite graphs: For some fixed
$0 < \alpha \leq 1$, a bipartite graph $G = (A, B, E)$ is called a
\emph{$(c, \alpha)$-expander} if
\begin{align*}
  \min_{\substack{{A' \subseteq A}\\{|A'| \leq \alpha |A|}}} \frac{|N(A')|}{|A'|} \geq c \enspace ,
\end{align*}
where $N(A') \subseteq B$ is the set of neighbours of $A'$ in $G$.
That is, in a $(c, \alpha)$-expander, every set  of vertices in
$A$ that is not too large is ``expanded'' by a factor $c$ by taking
one step in the graph.

Let $G = (A, B, E)$ be a random bipartite multigraph where
$|A| = |B| = n$ and where each vertex of $A$ is connected to three
vertices of $B$ chosen independently and uniformly at random (with
replacement). The following theorem shows that $G$ is an expander.
The proof of this theorem usually involves a messy sum that contains
binomial coefficients and probabilities: see, for example, 
Motwani and Raghavan~\cite[Theorem~5.3]{motwani.raghavan:randomized},
Pinsker~\cite[Lemma~1]{pinsker:expanders}, or 
Hoory, Linial, and Wigderson~\cite[Lemma~1.9]{hoory.linial.ea:expander}.

\begin{thm}
  There exists a constant $\alpha >0$ such that $G$ is a
  $(3/2,\alpha)$-expander with probability at least $1 - O(n^{-1/2})$.
\end{thm}

\begin{proof}
  We encode the graph $G$ by presenting its edge set. Since each edge
  is selected uniformly at random, the graph $G$ is chosen uniformly at
  random from a set of size $n^{3n}$.
  
  If $G$ is not a $(3/2, \alpha)$-expander, then there is some set $A'
  \subseteq A$ with $|A'|=k\le \alpha n$ and
  \[
  \frac{|N(A')|}{|A'|} < 3/2 \enspace .
  \]
  To encode $G$, we first give $k$ using an Elias $\gamma$-code; together 
  with the sets $A'$ and $N(A')$; and the edges between $A'$ and
  $N(A')$. Then we encode the rest of $G$, skipping the $3k\log n$
  bits devoted to edges incident to $A'$.  The key savings here come 
  because
  $N(A')$ should take $3k\log n$ bits to encode, but can actually be
  encoded in roughly $3k\log(3k/2)$ bits. Our code then has length
  \begin{align*}
    |C(G)| &= 2\log k + \log \binom{n}{k} + \log \binom{n}{3k/2} + 3k \log (3k/2) + (3n - 3k) \log n + O(1) \\
           &\le 2\log k + k\log n - k\log k + k\log e + (3k/2)\log n - (3k/2) \log (3k/2) \\
           & \hphantom{{}=2\log k} + (3k/2) \log e + 3k \log (3k/2) + (3n - 3k) \log n + O(1) \tag{by \eqref{log-n-choose-k}} \\
           &= 3n\log n - (k/2)\log n + (k/2)\log k + \beta k + 2\log k + O(1) \\
           &= \log n^{3n} - s(k)
  \end{align*}
  bits, where $\beta = (3/2) \log (3/2) + (5/2) \log e$ and
  \[
  s(k) = (k/2)\log n - (k/2)\log k - \beta k -
  2 \log k - O(1) \enspace.
  \]
  Since
  \[
    \frac{d^2}{dk^2} s(k) = \frac{4 - k}{2 k^2} \log e \enspace ,
  \]
  the function $s(k)$ is concave for all $k \geq 4$. Thus, $s(k)$ is
  minimized either when $k = 1, 2, 3, 4$, or when $k = \alpha n$. We have
  \[
    s(1) = (1/2)\log n + c_1,  \quad
    s(2) = \log n + c_2, \quad
    s(3) = (3/2) \log n + c_3, \quad 
    s(4) = 2 \log n + c_4, 
  \]
  for constants $c_1, c_2, c_3, c_4$. For $k=\alpha n$ we have
  \[
    s(\alpha n) = (\alpha n/2)\log \left(\frac{1}{2^{2\beta}
        \alpha}\right) - 2 \log \alpha n + c_5,
  \]
  for some constant $c_5$.
  Thus, $2^{-s(\alpha n)} = 2^{-\varOmega(n)}$ for
  $\alpha < (1/2)^{2\beta} \approx 0.002$. Now the Uniform Encoding
  Lemma gives the desired result. Indeed, the encoding works
  for all values of $k$, and it always saves at least 
  $s(1) = (1/2)\log n + O(1)$ bits.
  Thus, the construction fails with probability
  $O(n^{-1/2})$.
\end{proof}

\subsection{Permutations and Binary Search Trees}

We define a \emph{permutation} $\sigma$ of size $n$ to be a sequence
of $n$ pairwise distinct integers, sometimes denoted by
$\sigma = (\sigma_1, \dots, \sigma_n)$. The set
$\{\sigma_1, \dots, \sigma_n\}$ is called the
\emph{support} of $\sigma$. This slightly unusual
definition will serve us for the purpose of encoding. Except when
explicitly stated, we will assume that the support of a permutation
of size $n$
is precisely $\{1, \dots, n\}$. For any fixed support of
size $n$, the number of distinct permutations 
is $n!$.

\subsubsection{Analysis of Insertion Sort}
\seclabel{insertion-sort}

Recall the insertion sort algorithm for sorting a list
$\sigma = (\sigma_1,\ldots,\sigma_n)$ of $n$ elements:

\noindent{$\textsc{InsertionSort}(\sigma)$}
\begin{algorithmic}[1]
  \FOR{$i\gets 2$ \TO $n$}
     \STATE{$j \gets i$}
     \WHILE{$j>1$ \AND $\sigma_{j-1} > \sigma_j$}
         \STATE{$\sigma_j \leftrightarrow \sigma_{j-1}$
            \COMMENT{ swap }}\label{alg:isort:swap}
         \STATE{$j\gets j-1$}
     \ENDWHILE
  \ENDFOR
\end{algorithmic}

A typical task in the average-case analysis of algorithms is to
determine the number of times Line~\ref{alg:isort:swap} executes 
if $\sigma$ is a
uniformly random permutation of size $n$.  The answer
$\binom{n}{2}/2$ is an easy application of linearity of expectation:
For every one of the $\binom{n}{2}$ pairs of indices 
$p,q\in\{1,\dots,n\}$ with
$p<q$, the values initially stored at positions $\sigma_p$ and
$\sigma_q$ will eventually be swapped if and only if
$\sigma_p > \sigma_q$. This happens with probability $1/2$ in a
uniformly random permutation. A pair $p, q \in \{1, \dots, n\}$ with
$p < q$ and $\sigma_p > \sigma_q$ is called an
\emph{inversion}, so the number of times Line~\ref{alg:isort:swap} 
executes is the number
of inversions of $\sigma$.

A more advanced question is to ask for a concentration result on the
number of inversions. This is harder to tackle; because $>$
is transitive, the $\binom{n}{2}$ events being studied have a lot of
interdependence. In the following, we show how an encoding argument
leads to a concentration result.  The argument
presented here follows the same outline as Vit\'{a}nyi's analysis of
bubble sort \cite{vitanyi:analysis}, though without all the trappings
of Kolmogorov complexity.

\begin{thm}\thmlabel{insertion-sort}
  Let $\alpha \in (0, 1/e^2)$.
  A uniformly random permutation $\sigma$ of size $n$ has at most
  $\alpha n^2 - n + 2$ inversions with probability at most
  $2^{n \log(\alpha e^2) + O(\log n)}$. In particular, for a fixed
  $\alpha < 1/e^2$, this probability is $2^{-\varOmega(n)}$.
\end{thm}

\begin{proof}
  We encode the permutation $\sigma$ by recording the execution of
  \textsc{InsertionSort} on $\sigma$. In particular, we record for
  each $i\in\{2,\ldots,n\}$, the number of times $m_i$ that 
  Line~\ref{alg:isort:swap}
  executes during the $i$-th iteration of
  $\textsc{InsertionSort}(\sigma)$. With this information, one can run
  the following algorithm to recover $\sigma$:
  
  \noindent{$\textsc{InsertionSortReconstruct}(m_2, \dots, m_n)$}:
  \begin{algorithmic}[1]
    \STATE{$\sigma \gets (1, \dots, n)$}
    \FOR{$i \gets n \textbf{ down to } 2$}
      \FOR{$j \gets i - m_i + 1$ \TO $i$}
        \STATE{$\sigma_j \leftrightarrow \sigma_{j-1}$
          \COMMENT{ swap }}
      \ENDFOR
    \ENDFOR
    \RETURN{$\sigma$}
  \end{algorithmic}

  We have to be slightly clever with the
  encoding. Rather than encode $m_2, m_3, \dots, m_n$ directly, we first
  encode $m=\sum_{i=2}^{n} m_i$ using $2\log n$ bits (since
  $m < n^2$). Given $m$, it remains to describe the partition of
  $m$ into $n-1$ non-negative integers $m_2,\ldots,m_n$; there are
  $\binom{m+n-2}{n-2}$ such partitions.\footnote{To see this, draw
    $m+n-2$ white dots on a line, then choose $n-2$ dots to colour
    black. This splits the remaining $m$ white dots into $n-1$
    groups whose sizes determine the values of $m_2,\ldots,m_n$.}
  
  Therefore, the values of $m_2,\ldots,m_n$ can be encoded using
  \[
    |C(\sigma)| = 2\log n + \log\binom{m+n-2}{n-2}
  \]
  bits and this is sufficient to recover the permutation $\sigma$.  By
  applying \eqref{log-n-choose-k}, we obtain
  \begin{align*}
    |C(\sigma)| & \le (n-2)\log(m+n-2) - (n-2)\log(n-2)  + (n-2)\log e + O(\log n) \\
      & \le n\log(m+n-2) - n\log n   + n\log e + O(\log n) \\
      & \le n\log(\alpha n^2) - n\log n  + n\log e + O(\log n) \tag{since $m \le \alpha n^2 - n + 2$}\\
      & = 2n\log n + n\log\alpha - n\log n  + n\log e + O(\log n) \\
      & = n\log n + n\log\alpha + n\log e + O(\log n) \\
      & = \log n! + n\log\alpha + 2n\log e + O(\log n) \tag{by \eqref{stirling-loose}} \\
      & = \log n! + n \log (\alpha e^2) + O(\log n) \enspace .
  \end{align*}
  Again, we finish by applying the Uniform Encoding Lemma.
\end{proof}

\begin{rem}
  \thmref{insertion-sort} is not sharp; it only gives a non-trivial
  probability when $\alpha < 1/e^2$.  To obtain a sharp bound, one can
  use the fact that $m_2,\ldots,m_n$ are independent and that $m_i$ is
  uniform over $\{0,\ldots,i-1\}$ together with the method of bounded
  differences \cite{mcdiarmid:on}. This shows that $m$ is concentrated in
  an interval of size $O(n^{3/2})$.
\end{rem}

\subsubsection{Records}
\seclabel{records}

A \emph{(max) record} in a permutation $\sigma$ of size
$n$ is some value $\sigma_i$, $1 \leq i \leq n$, such that
\[
  \sigma_i = \max\{\sigma_1, \dots, \sigma_i\} \enspace .
\]
If $\sigma$ is chosen uniformly at random, the probability that 
$\sigma_i$ is a record is exactly $1/i$. Thus, 
the expected number of records in such a permutation is
\[
  H_n = \sum_{i = 1}^n 1/i = \ln n + O(1) \enspace ,
\]
the $n$-th harmonic number. It is harder to establish concentration
with non-negligible probability. To do this, one first needs to show
the independence of certain random variables, which quickly becomes
tedious. We instead give an encoding argument to show concentration of
the number of records, inspired by a technique used by 
Lucier, Jiang, and Li~\cite{lucier.jiang.li:quicksort}
 to study the height of random binary search trees
 (see also \secref{height}).

First, we describe a recursive encoding of a
permutation $\sigma$ of size $n$: Begin by providing the first value
of the permutation $\sigma_1$; then show the set of indices from
$\{2, \dots, n\}$
for which $\sigma$ takes on a value strictly smaller than
$\sigma_1$ and an explicit encoding of the induced permutation
on the elements at those indices; finally, give a recursive encoding of the
permutation induced on the elements strictly larger than
$\sigma_1$. The number of recursive invocations is equal to the number
of records in $\sigma$.

If $\sigma$ contains $k$ elements strictly smaller than $\sigma_1$,
then the length $\ell(\sigma)$ of the codeword for $\sigma$ satisfies 
\[
  \ell(\sigma) = \log n + \log \binom{n - 1}{k} + 
    \log k! + \ell(\sigma') \enspace ,
\]
where $\sigma'$ is the induced permutation on the $n - k - 1$ 
elements strictly
larger than $\sigma_1$.
Thus, we get the following recursion for the length $\ell(n)$ of the
encoding for a permutation of size $n$:
\[
  \ell(n) = \max_{k \in \{1, \dots, n-1\}} \left(\log n + 
  \log \binom{n - 1}{k} + 
    \log k! + \ell(n-1-k)\right) \enspace ,
\]
with $\ell(0) = 0$ and $\ell(1) = 0$. 
This solves to $\ell(n) = \log
n!$, so the encoding described above is no better than a fixed-length
encoding for $\sigma$. However, a simple modification of the scheme 
yields a result about the concentration of records in a uniformly
random permutation.

\begin{thm}\thmlabel{records}
  For any fixed $c > 2$, 
  a uniformly random permutation $\sigma$ of size $n$ has at least $c
  \log n$ records with probability at most
  \[
    2^{-c (1 - H(1/c)) \log n + O(\log \log n)}\enspace .
  \]
\end{thm}
\begin{proof}
We describe an encoding scheme for permutations with at least 
$t = \ceil{c \log n}$ records. 
Suppose that the permutation $\sigma$ has $t$ records
$r_1 < r_2 < \cdots < r_t$. First, we define a
bit string $x = (x_1, \dots, x_t) \in \{0, 1\}^t$, where $x_1 = 0$
and $x_i = 1$ if and only if $r_i$ lies in the second half of the
interval $[r_{i - 1}, n]$, for $i = 2, \dots, t$. 
Recalling that  $n_1(x)$ represents the 
number of ones in the bit string $x$, it follows that 
$n_1(x) \leq \log n$, so $n_1(x)/t \leq 1/c$.

To begin our encoding of $\sigma$, we encode the bit string $x$ by
giving the set of $n_1(x)$ ones in $x$; followed by the 
recursive encoding of $\sigma$ from earlier. Now, our knowledge of
the value of $x_i$ halves the size of the space of options for
encoding the position $r_i$ . In other words, our knowledge of $x$
allows us to encode each record using roughly one less bit per
record. More precisely, if the number of choices for each record
$r_i$ in the original encoding is $m_i$, such that $m_1 > \dots >
  m_t$, then the number of bits spent encoding records in the new code
  is at most
  \begin{align*}
    \sum_{i = 1}^t \log \ceil{m_i/2} &\leq \sum_{i = 1}^t \log (m_i/2
                                       + 1)
    \\
                                     &\leq \sum_{i = 1}^t \log (m_i/2) + \sum_{i = 1}^t O(1/m_i) \tag{since $\log (x + 1) = \log x + O(1/x)$} \\
                                     &\leq \sum_{i = 1}^t \log (m_i/2) + O(H_t) 
                                     = \sum_{i = 1}^t \log
                                       m_i - t + O(\log \log n) \enspace ,
  \end{align*}
  since $c$ is a constant. Thus, the total length of the code is
  \begin{align*}
    |C(\sigma)| &\leq \binom{t}{n_1(x)} + \log n! - t + O(\log \log n) \\
                &\leq \log n! - t (1 - H(n_1(x)/t)) + O(\log \log n) \tag{by \eqref{log-n-choose-k}}\\
                &\leq \log n! - c (1 - H(1/c))\log n + O(\log \log n) \enspace ,
  \end{align*}
  where this last inequality follows since $c > 2$, so 
  $0 \leq n_1(x)/t \leq 1/c < 1/2$, and $H(n_1(x)/t)
  \leq H(1/c)$ since $H(\cdot)$ is increasing on $[0, 1/2]$. 
  We finish by applying the Uniform Encoding Lemma.
\end{proof}

\begin{rem}\remlabel{records}
  The preceding result only works for $c > 2$, but it is  known
  that the number of records in a uniformly random permutation is
  concentrated around $\ln n + O(1)$, where $\ln n = \alpha \log n$
  for $\alpha = 0.6931\dots$. We leave as an open problem whether or not
  this significant gap can be closed through an encoding argument.
\end{rem}

\subsubsection{The Height of a Random Binary Search Tree}
\seclabel{height}

Every permutation $\sigma$ determines a binary search tree
$\text{BST}(\sigma)$ created through the sequential insertion of the
keys $\sigma_1, \ldots, \sigma_n$. Specifically, if $\sigma^L$
(respectively, $\sigma^R$) denotes the permutation of elements
strictly smaller (respectively, strictly larger) than $\sigma_1$, then
$\text{BST}(\sigma)$ has $\sigma_1$ as its root, with
$\text{BST}(\sigma^L)$ and $\text{BST}(\sigma^R)$ as left and right
subtrees.

Lucier, Jiang, and Li~\cite{lucier.jiang.li:quicksort} use an encoding
argument via Kolmogorov complexity to study the height of
$\text{BST}(\sigma)$. They show that for a uniformly chosen
permutation $\sigma$, the tree $\text{BST}(\sigma)$ has height at most
$c \log n$ with probability $1 - O(1/n)$ for $c = 15.498\dots$; we can
extend our result on records from~\secref{records} to obtain a 
tighter result.

For a node $u$, let $s(u)$ denote the number of nodes in the tree
rooted at $u$. Then, $u$ is called \emph{balanced} if
$s(u^L), s(u^R) > s(u)/4$, where $u^L$ and $u^R$ are the left and
right subtrees of $u$, respectively. In other words, since each node
$u$ determines an interval $[v, w]$, where $v$ is the smallest node in
the subtree rooted at $u$, and $w$ is the largest such node, then $u$
is balanced if and only if
\[
  u \in \left(\frac{w + v}{2} - \frac{w - v - 1}{4}, \frac{w + v}{2} + \frac{w - v - 1}{4}\right) \enspace ,
\]
\emph{i.e.}~$u$ is called balanced if it occurs inside the middle interval
of length $(w - v - 1)/2$ of its subrange.

\begin{thm}\thmlabel{bst-height}
  Let $\sigma$ be a uniformly random permutation of size $n$. 
  There is a constant $c < 9.943483$ such that 
  $\text{BST}(\sigma)$ has height at most $c\log n$ with probability
  $1 - O(1/n)$.
\end{thm}
\begin{proof}
  Let $c > 2/\log (4/3)$,
  and suppose that the tree $\text{BST}(\sigma)$ contains a path
  $Y = (y_1, \ldots, y_t)$ of length $t = \ceil{c \log n}$ that
  starts at the root and in which $y_{i+1}$ is a child of $y_i$,
  for $i = 1, \dots, t-1$.

  Our encoding for $\sigma$ has three parts. 
  The first part consists of
  a bit string
  $x = (x_1, \dots, x_t)$, where $x_i = 1$ if and only if $y_i$ is
  balanced. From our definition, if $y_i$ is balanced, then
  $s(y_{i + 1}) \leq (3/4) s(y_i)$. Since $n_1(x)$ counts the
  number of balanced nodes along $Y$, we get
  \[
    1 \leq (3/4)^{n_1(x)} n \iff n_1(x) \leq \log_{4/3} n \enspace .
  \]
  Next, our encoding contains a fixed-length encoding of $y_t$ using
  $\log n$ bits.

  The third part of our encoding is recursive: First, encode the value of
  the root $y_1$ using $\log \ceil{n/2}$ bits. Note that since we know
  whether $y_1$ is balanced or not, there are only $n/2$ possibilities
  for the root value, by the discussion above.  If $y_2$ is the left child
  $y_1$, then specify the values in the right subtree of $y_1$, including
  an explicit encoding of the permutation induced by these values; and
  recursively encode the permutation of values in subtree of $y_2$. If,
  instead, $y_2$ is the right child of  $y_1$, proceed symmetrically.
  (Note that a decoder can determine which of these two cases occured by
  comparing $y_t$ with $y_1$ since $y_2 <y_1$ if and only if $y_t< y_1$.)
  Once we reach $y_t$, we encode the permutations of the two subtrees
  of $y_t$ explicitly.
 
  The first two parts of our encoding use at most
  \[
    t H(n_1(x)/t) + \log n
  \]
  bits. The same analysis as in the proof of \thmref{records} shows that
  the second part of our encoding has length at most
  \[
    \log n! - t + O(\log \log n) \enspace .
  \]
  In total, our code has length
  \begin{align*}
    |C(\sigma)| &= \log n! - t + t H(n_1(x)/t) + \log n + O(\log \log n) \\
                &\le \log n! - c \log n + c \log n H\left(\frac{1}{c \log (4/3)}\right) + \log n + O(\log \log n) \\
                &= \log n! - c \left(1 - H\left(\frac{1}{c \log (4/3)}\right)\right) + \log n + O(\log \log n) \enspace ,
  \end{align*}
  where the inequality uses the fact that $c > 2/\log(4/3)$. Applying
  the Uniform Encoding Lemma, we see that $\text{BST}(\sigma)$ has
  height at most $c \log n$ with probability $1 - O(1/n)$ for
  $c > 2/\log (4/3)$ satisfying
  \[
    c \left(1 - H\left(\frac{1}{c \log (4/3)}\right)\right) > 2 \enspace ,
  \]
  and a computer-aided calculation shows that 
  $c = 9.943483$ 
  satisfies this inequality.
\end{proof}

\begin{rem}
  Devroye, Morin, and Viola~\cite{devroye:records} show 
  how the length of the path to
  the key $i$ in $\text{BST}(\sigma)$ relates to the number of records
  in $\sigma$. Specifically, he notes that the number of records in
  $\sigma$ is the number of nodes along the rightmost path in
  $\text{BST}(\sigma)$. Since the height of a tree is the length of
  its longest root-to-leaf path, we obtain as a corollary that the
  number of records in a uniformly random permutation is $O(\log n)$
  with high probability; the result from \thmref{records} only
  improves upon the implied constant.
\end{rem}

\begin{rem}
  We know that the height of the binary search tree built from the
  sequential insertion of elements from a uniformly random permutation
  of size $n$ is concentrated around $\alpha \ln n + O(\log \log n)$,
  for $\alpha = 4.311\dots$~\cite{reed:height}. Perhaps if the gap in
  our analysis of records in \remref{records} can be closed through an
  encoding argument, then so too can the gap in our analysis of random
  binary search tree height.
\end{rem}

\subsubsection{Hoare's Find Algorithm}

In this section, we analyze the number of comparisons made in an
execution of Hoare's classic \textsc{Find} algorithm~\cite{hoare:find}
which returns the $k$-th smallest element in an array of $n$
elements. The analysis is similar to that of the preceding section.

We refer to an easy algorithm $\textsc{Partition}$, which takes as
input an array $\sigma = (\sigma_1, \dots, \sigma_n)$ and partitions
it into the arrays $\sigma^L$ and $\sigma^R$ which contain the values
strictly smaller and strictly larger than $\sigma_1$,
respectively. The element $\sigma_1$ is called a \emph{pivot}. The
algorithm $\textsc{Partition}$ can be implemented so as to perform
only $n - 1$ comparisons as follows:

\noindent{$\textsc{Partition}(\sigma)$}:
\begin{algorithmic}[1]
  \STATE{$\sigma^L, \sigma^R \gets \textbf{nil}$}
  \FOR{$i \gets 2$ \TO $n$}
    \IF{$\sigma_i > \sigma_1$}
      \STATE{push $\sigma_i$ onto $\sigma^R$}
    \ELSE
      \STATE{push $\sigma_i$ onto $\sigma^L$}
    \ENDIF
  \ENDFOR
  \RETURN{$\sigma^L, \sigma^R$}
\end{algorithmic}

Using this, we give the algorithm $\textsc{Find}$:

\noindent{$\textsc{Find}(k, \sigma)$}:
\begin{algorithmic}[1]
  \STATE{$\sigma^L, \sigma^R \gets \textsc{Partition}(\sigma)$}
  \IF{$|\sigma^L| \geq k$}
    \RETURN{$\textsc{Find}(k, \sigma^L)$}
  \ELSIF{$|\sigma^L| < k - 1$}
    \RETURN{$\textsc{Find}(k - |\sigma^L| - 1, \sigma^R)$}
  \ENDIF
  \RETURN{$\sigma_1$}
\end{algorithmic}

Suppose that the algorithm $\textsc{Find}$ sequentially identifies $t$
pivots $x_1, \dots, x_t$ before finding the solution. Let
$\sigma^{(i)}$ denote the value of $\sigma$ in the $i$-th recursive
call and let $n_i=|\sigma^{(i)}|$, so that
$\sigma^{(0)}=(\sigma_1,\ldots,\sigma_n)$ and $n_0=n$.  We will say
that the $i$-th pivot is \emph{good} if its rank, in $\sigma^{(i)}$,
is in the interval $[n_i/4, 3n_i/4]$. Note that a good pivot causes
the algorithm to recurse in a problem of size at most $3n_i/4$.

\begin{lem}
  Fix some constants $t_0 \geq 1$ and $0 < \alpha < 1/2$. Suppose
  that, for each $t_0 \leq i \leq t$, the number of good pivots among
  $x_1, \ldots, x_i$ is at least $\alpha i$. Then, $\textsc{Find}$
  makes $O(n)$ comparisons.
\end{lem}
\begin{proof}
  If $x_j$ is a good pivot, then the conditions of the lemma give that
  $n_j \le (3/4) n_{j - 1}$. Therefore, $n_i \leq (3/4)^{\alpha i} n$
  for each $t_0 \leq i \leq t$, and the total number of comparisons
  made by $\textsc{Find}$ is at most
  \[
    \sum_{i=0}^t n_i \le t_0 n + \sum_{i=t_0}^t n_i \le O(n) + n
    \sum_{i = t_0}^t (3/4)^{\alpha i} = O(n) \enspace . 
  \]
\end{proof}

\begin{thm}
  Let $\sigma$ be a uniformly random permutation. Then, for 
  every fixed probability $p \in (0,1)$, there exists
  a constant $c$ such that 
  $\textsc{Find}(k, \sigma)$ executes at most $cn$ comparisons with 
  probability at least $p$, for any $k$.
\end{thm}
\begin{proof}
  We again encode the permutation $\sigma$. Set $\alpha = 1/4$ and
  let $t_0$ be a constant depending on $p$. Suppose that the
  conditions of the preceding lemma are not satisfied for $\alpha$ and 
  $t_0$,
  \emph{i.e.}~there is an $i \geq t_0$ such that the number of good pivots
  among $x_1, \dots, x_i$ is less than $\alpha i$. We encode $\sigma$
  in two parts. The first part of our encoding gives the value of $i$
  using an Elias $\delta$-code, followed by the set of indices of the
  good pivots among $x_1, \dots, x_i$, which costs
  \[
  \log i + i H(\alpha) + O(\log \log i) \enspace .
  \]
  Note that the pivots $x_1, \dots, x_i$ trace a path from the root in
  $\text{BST}(\sigma)$. Therefore, the second part of our encoding is
  the recursive encoding presented in \secref{height}, in which each
  pivot can be encoded using one less bit, since knowing whether $x_j$
  is a good pivot or not narrows down the range of possible values for
  $x_j$ by a factor of $1/2$. In total, our code then has length
  \[
    |C(\sigma)| \le \log n! - i + i H(\alpha) + \log i + O(\log \log
    i) = \log n! - \varOmega(i) \enspace ,
  \]
  since $\alpha = 1/4 < 1/2$. The proof is completed by applying 
  the Uniform
  Encoding Lemma, and by observing that $t_0 \leq i$ can be made
  arbitrarily large.
\end{proof}

\subsection{$k$-SAT and the Lov\'{a}sz Local Lemma}

We now consider the question of satisfiability of propositional
formulas. Let us start with some definitions.

A (Boolean) \emph{variable} $x$ is either $\textsf{true}$ or
$\textsf{false}$. The negation of $x$ is denoted by $\neg x$. A
\emph{literal} is either a variable or its negation. A
\emph{conjunction} of literals is an ``and'' of literals, denoted by
$\land$. A \emph{disjunction} of literals is an ``or'' of literals,
denoted by $\lor$. A \emph{formula} $\varphi$ is an expression
including conjunctions and disjunctions of literals, and the set of
variables involved in this formula is called the \emph{support} of
$\varphi$. A \emph{clause} is a disjunction of literals, \emph{i.e.}
the ``or'' of a set of variables or their negations, \emph{e.g.}
\begin{align}
  x_1 \lor \neg x_2 \lor x_3 \enspace . \eqlabel{clause-example}
\end{align}
Two clauses will be said to intersect if their supports intersect. The
truth value which a formula $\varphi$ evaluates to under the
assignment of values $\alpha$ to its support will be denoted by
$\varphi(\alpha)$, and such a formula is said to be \emph{satisfiable}
if there exists an $\alpha$ with $\varphi(\alpha) = \textsf{true}$.
For example, the clause in \eqref{clause-example} is satisfied for all
truth assignments except
\[
  (x_1, x_2, x_3) = (\textsf{false}, \textsf{true}, \textsf{false})
  \enspace ,
\]
and indeed any clause is satisfied by all but one truth assignment
for its support. The formulas we are concerned with are conjunctions
of clauses, which are said to be in \emph{conjunctive normal form}
(CNF). More specifically, when each clause has at most $k$
literals, we call it a $k$-CNF formula.

The $k$-SAT decision problem asks to determine whether or not a given
$k$-CNF formula is satisfiable. In general, this problem is hard. Of
course, any satisfying truth assignment to the variables in a CNF
formula induces a satisfying truth assignment for each of its
clauses. Moreover, if the supports of the clauses are pairwise
disjoint, then the formula is trivially satisfiable,
and as we will see, this holds even if the 
clauses are only nearly pairwise disjoint, \emph{i.e.},  if for each
clause the support is disjoint from the supports of all but
less than $2^k/e$ other clauses.

This result has been well known as a consequence of the Lov\'{a}sz
Local Lemma~\cite{lovasz:locallemma}, whose original proof is
non-constructive, and so does not produce a satisfying truth
assignment (in polynomial time) when applied to an instance of
$k$-SAT. Some efficient constructive solutions to $k$-SAT have been
known, but only for suboptimal clause intersection sizes.
Moser~\cite{moser:ksat} first presented a constructive solution to
$k$-SAT with near optimal clause intersection sizes, and 
Moser and Tardos~\cite{moser:locallemma} then 
generalized this result to the
full Lov\'{a}sz Local Lemma for optimal clause intersection sizes. The
analysis which we reproduce in this section comes from Fortnow's
rephrasing of Moser's proof for $k$-SAT using the incompressibility
method~\cite{fortnow:ksat}.

Moser's algorithm is remarkably na\"\i ve, and can be described in only a
few sentences: Pick a uniformly random truth assignment for the
variables of $\varphi$. For each unsatisfied clause, attempt to fix it
by producing a new uniformly random truth assignment for its support,
and recursively fix any intersecting clause which is made unsatisfied
by this reassignment. We describe this process more carefully in the
algorithms $\textsc{Solve}$ and $\textsc{Fix}$ below.

\noindent{$\textsc{Solve}(\varphi)$}:
\begin{algorithmic}[1]
  \STATE{$\alpha \gets $ uniformly random truth assignment in $\{\textsf{true}, \textsf{false}\}^n$}
  \WHILE{$\varphi(\alpha) = \textsf{false}$}
    \STATE{$D \gets $ an unsatisfied clause in $\varphi$}
    \STATE{$\alpha \gets \textsc{Fix}(\varphi, \alpha, D)$}
  \ENDWHILE
  \RETURN{$\alpha$}
\end{algorithmic}

\noindent{$\textsc{Fix}(\varphi, \alpha, D)$}:
\begin{algorithmic}[1]
  \STATE{$\beta \gets $ uniformly random truth assignment in $\{\textsf{true}, \textsf{false}\}^k$}
  \STATE{replace the assignments in $\alpha$ for $D's$ support with the values in $\beta$}
  \WHILE{$\varphi(\alpha) = \textsf{false}$}
    \STATE{$D' \gets $ an unsatisfied clause in $\varphi$ intersecting $D$}
    \STATE{$\alpha \gets \textsc{Fix}(\varphi, \alpha, D')$}
  \ENDWHILE
  \RETURN{$\alpha$}
\end{algorithmic}

\begin{thm}
  Given a $k$-CNF formula $\varphi$ with $m$ clauses and $n$ variables
  such that each clause intersects at most $r < 2^{k - 3}$ other
  clauses, then the total number of invokations of $\textsc{Fix}$ in
  the execution of $\textsc{Solve}(\varphi)$ is at least
  $s + m \log m$ with probability at most $2^{-s}$.
\end{thm}
\begin{proof}
  Suppose that $\textsc{Fix}$ is called 
  $t = \ceil{s + m\log m}$ times. Let
  $\alpha \in \{\textsf{true}, \textsf{false}\}^n$ be the initial
  truth assignment for $\varphi$, and let
  $\beta_1, \ldots, \beta_t \in \{\textsf{true}, \textsf{false}\}^k$
  be the local truth assignments produced in each call to
  $\textsc{Fix}$. The string
  $\gamma = (\alpha, \beta_1, \ldots, \beta_t)$ is uniformly chosen
  from a set of size $2^{n + tk}$, and will be the subject of our
  encoding.

  The execution of $\textsc{Solve}(\varphi)$ determines a (rooted
  ordered) \emph{recursion tree} $T$ on $t + 1$ nodes as follows: The
  root of $T$ corresponds to the initial call to
  $\textsc{Solve}(\varphi)$. Every other node corresponds to a call
  to $\textsc{Fix}$. The children of a node correspond 
   to the sequence of calls to $\textsc{Fix}$ that the procedure
  performs, ordered from left to right.
  Each (non-root) node in the tree is assigned a clause
  and its uniformly random truth assignment produced during the call
  to $\textsc{Fix}$. Moreover, a pre-order traversal of this tree
  describes the order of function calls in the algorithm's execution.

  The string $\gamma$ can be recovered in a bottom-up manner from our
  knowledge of the tree $T$ and the final truth assignment
  $\alpha'$ after $t$ calls to \textsc{Fix}. 
  Specifically, let $D_1, \dots, D_t$ be the clauses
  encountered in a pre-order traversal of $T$. In particular, $D_t$ is
  the last fixed clause in the execution. Since $D_t$ was not
  satisfied before its reassignment, this allows us to deduce $k$
  values of the previous assignment
  before $D_t$ was fixed. Pruning $D_t$ from the tree and continuing
  in this manner at $D_{t - 1}$, we eventually recover the original
  truth assignment $\alpha$ produced in $\textsc{Solve}(\varphi)$.

  Therefore, to encode $\gamma$, we give the final truth
  assignment $\alpha'$; and a description of the shape of the tree
  $T$; and the sequence of at most $m$ clauses which are children of
  the root of $T$; and the at most $t$ clauses involved in the calls
  to $\textsc{Fix}$ in a pre-order traversal of $T$.  

  The key savings come from the fact that each clause intersects at
  most $r$ other clauses, so each clause (which is not a child of the
  root) can be encoded using $\log r$ bits. Each clause which is a
  child of the root can be encoded using $\log m$ bits, and since the
  order of these children might be significant, we use $m \log m$ bits
  to encode the full sequence of these clauses. Finally, as in
  \lemref{two-choice-component-size}, the shape of $T$ can be encoded
  using $2t$ bits. In total, the code has length
  \begin{align*}
    |C(\gamma)| &\le n + 2t + m\log m+ t \log r \\
    &\le n + 2t + m\log m + t(k - 3) \tag{since $r \le 2^{k - 3}$} \\
    &= n + tk - t + m\log m 
    \le n + tk - s \enspace .
  \end{align*}
  The result is obtained by applying the Uniform Encoding Lemma.
\end{proof}

\begin{rem}
  By more carefully encoding of the shape of the recursion tree above,
  Messner and Thierauf~\cite{messner:ksat} gave an encoding argument
  for the above result in which $r < 2^k/e$. Specifically, their 
  refinement follows
  from a more careful counting of the number of trees with nodes of
  bounded degree.
\end{rem}

\section{The Non-Uniform Encoding Lemma and Shannon-Fano Codes}
\seclabel{nuel}

Thus far, we have focused on applications that could always be
modelled as choosing some element $x$ uniformly at random from a
finite set $X$. To encompass even more applications, it is helpful to
have an Encoding Lemma that deals with \emph{non-uniform}
distributions over $X$. First, we recall the following useful classic
results: 
\begin{thm}[Markov's Inequality]
  For any non-negative random variable $Y$ with finite
  expectation, and any $a > 0$,
  \[
    \Pr\{Y \geq a\} \leq (1/a) \E\{Y\} \enspace .
  \]
\end{thm}

We will say that a real-valued function $\ell : X \to \R$
\emph{satisfies Kraft's condition} if
\[
  \sum_{x \in X} 2^{-\ell(x)} \leq 1 \enspace .
\]
\begin{lem}[Kraft's Inequality and prefix-free codes]
  If $C : X \nrightarrow \{0,1\}^*$ is a partial prefix-free code,
  then the function  $\ell : x \mapsto |C(x)|$ satisfies Kraft's condition. Conversely, for any
  function $\ell : X \to \N$ satisfying Kraft's condition, there
  exists a prefix-free code $C : X \to \{0, 1\}^*$ such that
  $|C(x)| = \ell(x)$ for all $x \in X$.
\end{lem}

The following generalization of the Uniform Encoding Lemma, which was
originally proven by Barron~\cite[Theorem~3.1]{barron:dissertation},
serves for non-uniform input distributions:
\begin{lem}[Non-Uniform Encoding Lemma]\lemlabel{nuel}  
  Let $C\from X\nrightarrow\{0,1\}^*$ be a partial prefix-free code,
  and let $p$ be a probability distribution on $X$.  Suppose we draw
  $x \in X$ randomly with probability $p_x$.  Then, for any
  $s \geq 0$,
  \[
    \Pr\{ |C(x)| \le \log(1/p_x)-s\} \le 2^{-s} \enspace .
  \]
\end{lem}

\begin{proof}
  We use Chernoff's trick of exponentiating both sides before applying
  the Markov inequality, and Kraft's inequality:
  \begin{align*}
    \Pr\{ |C(x)| \le\log(1/p_x)-s \}
    & = \Pr\{|C(x)| -\log(1/p_x) \le -s \} \\
    & = \Pr\{\log(1/p_x)-|C(x)| \ge s \} \\
    & = \Pr\left\{2^{\log(1/p_x)-|C(x)|} \ge 2^s \right\}  \tag{Chernoff's trick} \\
    & \le \frac{\E\left\{2^{\log(1/p_x)-|C(x)|}\right\}}{2^s} \tag{Markov's inequality} \\
    & = \frac{1}{2^s} \left(\sum_{x\in X} p_x\cdot 2^{\log(1/p_x)-|C(x)|}\right) \\
    & = \frac{1}{2^s} \left(\sum_{x\in X}2^{-|C(x)|}\right) \enspace .
  \end{align*}
  By Kraft's inequality,
  $\sum_{x \in X} 2^{-|C(x)|} \leq 1$, and the result
  is obtained.
\end{proof}

The Non-Uniform Encoding Lemma is a strict generalization
of the Uniform Encoding Lemma: Take $p_x=1/|X|$ for all $x\in X$ and
we obtain the Uniform Encoding Lemma.

As in \secref{sparse-bit-strings}, we will be interested in using a
Shannon-Fano code $C_\alpha$ to encode $\mathrm{Bernoulli}(\alpha)$
bit strings of length $n$. Recall that for such a string $x$, this
code has length
\[
  n_1(x) \log (1/\alpha) + n_0(x) \log(1/(1 - \alpha)) \enspace ,
\]
since for now we are not concerned with ceilings.

\section{Applications of the Non-Uniform Encoding Lemma}
\seclabel{applications-ii}

\subsection{Chernoff Bound}
\seclabel{chernoff}

We will now prove the so-called \emph{additive version} of the
Chernoff bound on the tail of a binomial random
variable~\cite{chernoff:bound}. \thmref{chernoff-basic} dealt with
the special case of this result for $\mathrm{Bernoulli}(1/2)$ bit
strings.

\begin{thm}\thmlabel{chernoff}
  If $B$ is a $\mathrm{Binomial}(n,p)$ random variable, then for any
  $\eps \geq 0$,
  \[
    \Pr\{B\le(p-\eps)n\} \le 2^{-nD(p-\eps \,\|\, p)} \enspace ,
  \]
  where 
  \[ 
    D(p\, \|\, q)= p\log (p/q) + (1-p)\log ((1 - p)/(1 - q))
  \]
  is the \emph{Kullback-Leibler divergence} or \emph{relative entropy}
  between $\mathrm{Bernoulli}(p)$ and $\mathrm{Bernoulli}(q)$ random
  variables.
\end{thm}

\begin{proof}
  By definition, $B=\sum_{i=1}^n x_i$, where $x_1,\ldots,x_n$ are
  independent $\mathrm{Bernoulli}(p)$ random variables.  We will use
  an encoding argument on the bit string $x=(x_1,\ldots,x_n)$. The
  proof is almost identical to that of \thmref{chernoff-basic}---now,
  we encode $x$ using a Shannon-Fano code $C_\alpha$, with
  $\alpha = p - \eps$. Such a code has length
  \[
    |C_{p - \eps}(x)| = n_1(x) \log (1/(p - \eps)) + n_0(x) \log (1/(1
    - p + \eps)) \enspace .
  \]
  Now, $x$ appears with probability
  $p_x = p^{n_1(x)} (1 - p)^{n_0(x)}$. This allows us
  to express $|C_{p - \eps}(x)|$ in terms of $\log(1/p_x)$ as follows
  \begin{align*}
    |C_{p - \eps}(x)| &= \log(1/p_x) + 
         n_1(x) \log (p/(p - \eps)) + (n - n_1(x)) \log ((1 - p)/(1 - p + \eps)) \\
         &= \log (1/p_x) + n_1(x) \log \left(1 + \frac{\eps}{p - \eps}\right) + 
(n_1(x) - n) \log \left(1 + \frac{\eps}{1 - p}\right) \enspace ,
  \end{align*}
  and $|C_{p - \eps}(x)|$ increases as a function of
  $n_1(x)$. Therefore, if $n_1(x) \le (p - \eps)n$, then
  \begin{align*}
    |C_{p - \eps}(x)| &\le \log (1/p_x) - n (p - \eps) \log ((p - \eps)/p) - n (1 - p + \eps) \log ((1 - p + \eps)/(1 - p)) \\
                      &= \log (1/p_x) - n D(p - \eps \,\|\, p) \enspace .
  \end{align*}
  The Chernoff bound is obtained by applying the Non-Uniform Encoding
  Lemma.
\end{proof}

\subsection{Percolation on the Torus}

Percolation theory studies the emergence of large components in random
graphs. For a general study of percolation theory, see the book by
Grimmett~\cite{grimmett:percolation}.  We give an encoding argument
proving that percolation occurs on the torus when edge survival rate
is greater than $2/3$, \emph{i.e.}~in random subgraphs of the torus
grid graph in which each edge is included independently at random with
probability at least $2/3$, only at most one large component
emerges. Our line of reasoning follows what is known as a
\emph{Peierls argument}.

Suppose that $\sqrt{n}$ is an integer. The \emph{$\sqrt{n} \times \sqrt{n}$
  torus grid graph} is defined to be the graph with vertex set
$\{1, \ldots, \sqrt{n}\}^2$, where $(i, j)$ is adjacent to $(k, l)$
if
\begin{itemize}[topsep=0pt]
\item $|i - k| \equiv 1 \pmod{\sqrt{n}}$ and $j = l$, or
\item $|i - k| = 0$ and $|j - l| \equiv 1 \pmod{\sqrt{n}}$.
\end{itemize}

\begin{thm}\thmlabel{percolation-long-cycles}
  Suppose that $\sqrt{n}$ is an integer.  Let $G$ be a subgraph of the
  $\sqrt{n} \times \sqrt{n}$ torus grid graph in which each edge is
  chosen with probability $p < 1/3$. Then, the probability that $G$
  contains a cycle of length at least
  \[
    \frac{s + \log n}{\log (1/(3p))}
  \]
  is at most $2^{-s}$.
\end{thm}
\begin{proof}
  Let $A$ be the bit string of length $2n$ encoding the edge set of
  $G$. then the probability $p_G$ that the graph $G$ is sampled is
  \[
     p_G  = p^{n_1(A)}(1-p)^{n_0(A)} \enspace .
  \]
  Suppose that $G$ contains a cycle $C'$ of length
  $t \geq (s + \log n + O(1))/\log (1/(3p))$. We encode $A$ as follows:
  first, we give a single vertex $u$ in $C'$ ($\log n$ bits). Then,
  we provide the sequence of directions that the cycle
  moves along from $u$. 
  There are four possibilities for the direction of the
  first step taken by $C'$ from $u$, but only three for each
  subsequent choice. Thus, this sequence can be specified by
  $2 + (t - 2) \log 3$ bits.
  We conclude with a Shannon-Fano code with parameter $p$ for
  the remaining edges of $G$ ($(n_1(A) - t) 
  \log (1/p) + n_0(A) \log (1/(1 - p)$)
  bits. The total length of our code is then
  \begin{align*}
    |C(G)| &= \log n + 2 + (t - 2) \log 3 + (n_1(A) - t) \log (1/p) +
             n_0(A) \log (1/(1 - p)) \\
           &\leq \log (1/p_G) + \log n - t \log (1/(3p)) + O(1) \\
           &\leq \log (1/p_G) - s
  \end{align*}
  by our choice of $t$, since $2 +  (t-2) \log 3 \leq t\log 3$.
  We finish by applying the Non-Uniform Encoding
  Lemma.
\end{proof}

The torus grid graph can be drawn in the obvious way without crossings
on the surface of a torus. This graph drawing gives rise to a dual
graph, in which each vertex corresponds to a face in the primal
drawing, and two vertices are adjacent if and only their primal faces
are incident to the same edge. This dual graph is isomorphic to the
original torus grid graph.

This drawing of the torus grid graph also induces drawings for
any of its subgraphs. Any such subgraph also has a dual, where each
vertex corresponds to a face in the dual torus grid graph, and two
vertices are adjacent if and only if their corresponding faces are
incident to the same edge of the original subgraph.

\begin{thm}
  Suppose that $\sqrt{n}$ is an integer.  Let $G$ be a subgraph of the
  $\sqrt{n} \times \sqrt{n}$ torus grid graph in which each edge is
  chosen with probability greater than $2/3$. Then, $G$ has at most
  one component of size $\omega(\log^2 n)$ with high probability.
\end{thm}
\begin{proof}
  See \figref{peierls} for a visualization of this phenomenon. Suppose
  that $G$ has at least two components of size $\omega(\log^2
  n)$. Then, there is a cycle of faces separating these components
  whose length is $\omega(\log n)$. From the discussion above, such a
  cycle corresponds to a cycle of $\omega(\log n)$ missing edges in
  the dual graph, as in \figref{peierlsrare}. From
  \thmref{percolation-long-cycles}, we know that this does not happen
  with high probability.
\end{proof}

\begin{figure}
  \centering
  \begin{subfigure}[t]{0.45\textwidth}
    \includegraphics[scale=0.9]{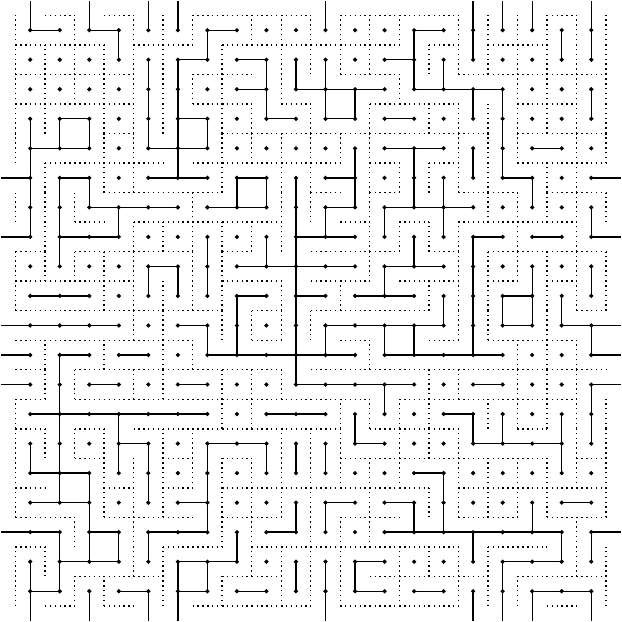}
    \caption{When $p = 0.33 < 1/3$, long cycles are rare. Dotted lines show
      missing edges in the dual.}
    \figlabel{peierlsrare}
  \end{subfigure}
  \quad\quad
  \begin{subfigure}[t]{0.45\textwidth}
    \includegraphics[scale=0.9]{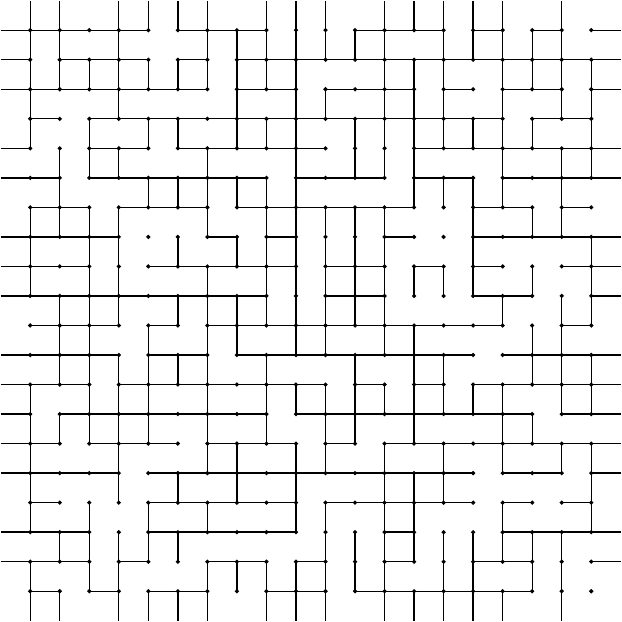}
    \caption{When $p = 0.67 > 2/3$, there is likely to be only one large
      component.}
  \end{subfigure}
  \caption{Random subgraphs of the $20 \times 20$ torus grid graph.}
  \figlabel{peierls}
\end{figure}

\subsection{Triangles in $G_{n,p}$}

Recall as in \secref{ramsey} that the Erd\H{o}s-R\'{e}nyi random graph
$G_{n,p}$ is the probability space of undirected graphs with vertex set
$V=\{1,\ldots,n\}$ and in which each edge $\{u, w\} \in \binom{V}{2}$
is present with probability $p$ and absent with probability $1-p$,
independently of the other edges.

By linearity of expectation, the expected number of triangles (cycles of 
length 3) 
in $G_{n,p}$ is $p^3\binom{n}{3}$.  For $p=(6c)^{1/3}/n$, this
expectation is $c-O(1/n)$.  Unfortunately, even when $c$ is a large
constant, it still takes some work to show that there is a constant
probability that $G_{n,p}$ contains at least one triangle. Indeed,
this typically requires the use of the second moment method, which
involves computing the variance of the number of triangles in
$G_{n,p}$. To show that $G_{n, p}$ has a triangle with more
significant probability is even more complicated,
 and a proof of this result would still typically rely on an
advanced probabilistic inequality~\cite{alon:probabilistic}. Here we
show how this can be accomplished with an encoding argument.

\begin{thm}\thmlabel{triangles-up}
  For $c \in (0, \log^{1/3} n]$ and 
  $p=c/n$, $G \in G_{n,p}$ contains at least one 
  triangle with
  probability at least $1-2^{-\varOmega(c^3)}$.
\end{thm}

\begin{proof}
  In this argument, we will produce an encoding of $G$'s adjacency
  matrix, $A$. For simplicity of exposition, we assume that $n$ is
  even.

  Refer to \figref{triangles}.  If $G$ contains no triangles, then we
  look at the number of ones in the $n/2\times n/2$ submatrix $M$
  determined by rows $1,\ldots,n/2$ and columns $n/2+1,\ldots,n$. Note
  that $n_1(M)$, the number of ones in $M$, is a
  $\mathrm{Binomial}(n^2/4, c/n)$ random variable with expectation
  $cn/4$.  There are three
  events to consider:
  
  \begin{figure}
    \centering{\includegraphics{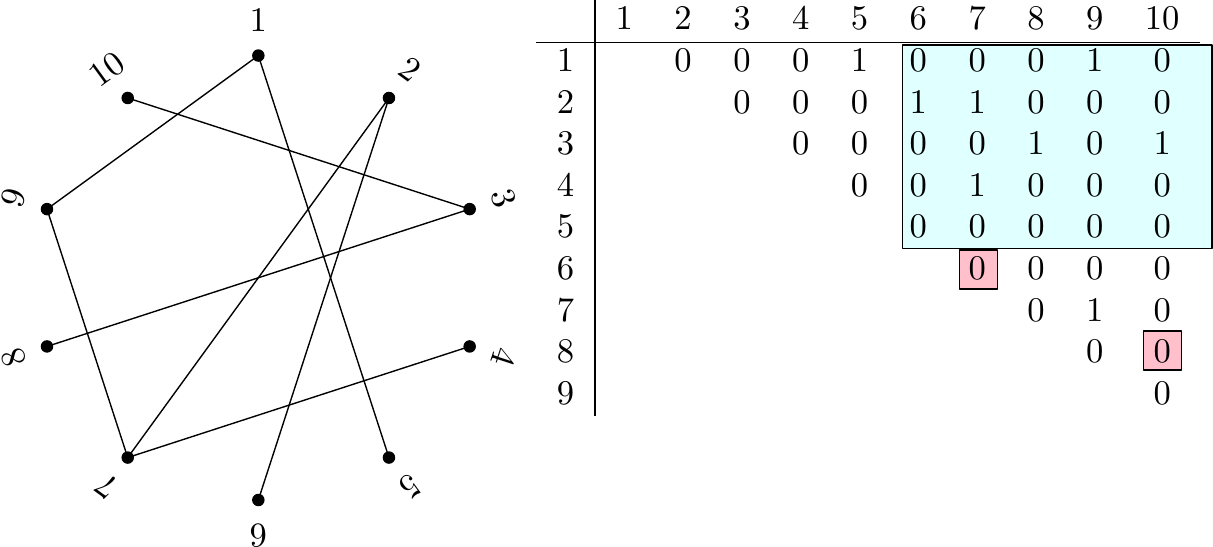}}
    \caption{The random graph $G_{n,c/n}$ contains triangles when $c$
      is large enough.  The highlighted 0 bits in the last five rows
      can be deduced from pairs of 1 bits in the first 5 rows.}
    \figlabel{triangles} 
  \end{figure}

  \begin{enumerate}
  \item \textbf{Event $\mathcal{E}_1$}: The number $n_1(M)$ of 
  ones in $M$ is at most $cn/8$.  In this case,
    the number of ones in this submatrix is much less than the
    expected number, $cn/4$.  
    We can simply apply
    Chernoff's bound to show that $\Pr \left\{ \mathcal{E}_0 \right\}
    = 2^{-\varOmega(c^3)}$.
    We leave this as an exercise to the reader.

  \item \textbf{Event $\mathcal{E}_2$}:
    Fix $n/2 + 1 \leq j <  k \leq n$, and let $\mathcal{B}_{jk}$ 
    be the event that
    $M$ contains at least three rows $1 \leq i_1 < i_2 < i_{3}
    \leq n/2$ with $A_{i_l, j} = A_{i_l, k} = 1$, for $l = 1, 2, 3$.
    Clearly, we have $\Pr \left\{\mathcal{B}_{jk}\right\} 
    \leq \binom{n/2}{3} (c/n)^{6}
    = O\left(c^{6}/n^{3}\right)$. Let $\mathcal{E}_2$ be the event that
    $\mathcal{B}_{jk}$ holds for at least one pair 
    $n/2+1 \leq j < k \leq n$.
    Since there are $\binom{n/2}{2}$ such pairs, we have
    $\Pr\left\{\mathcal{E}_2\right\} = O(c^{6}/n) = 
    2^{-\varOmega(c^3)}$, as we
    assumed $c \leq \log^{1/3} n$.
  \item \textbf{Event $\mathcal{E}_3$}: Let $\mathcal{E}_3$ 
  be the event that (i) 
  the number of ones in $M$ 
   is greater than $cn/8$; (ii)  for each pair $n/2+1 \leq j < k \leq n$,
   there are at most two rows in $M$ where both the entry with
   index $j$ and the entry with index $k$ are set two one; and (iii) $G$
   contains no triangles.
    Notice that, for $i<j<k$ if $A_{i,j}=1$ and $A_{i,k}=1$, then the
    fact that there are no triangles implies that $A_{j,k}=0$.

    Let $m_i$ be the number of ones in the $i$-th row of the
    submatrix.  By specifying rows $1,\dots,n/2$, we eliminate the
    need to specify
    \[
      m \geq (1/2)\sum_{i=1}^{n/2}\binom{m_i}{2} \ge (n/4) \binom{2
        n_1(M)/n}{2} \ge (n/4)\binom{c/4}{2} = \varOmega(c^2n) \enspace ,
    \]
    zeros in rows $n/2+1,\dots,n$ (here, we used the fact that each
    pair of ones appears in at most two rows and that the
    function $x \mapsto \binom{x}{2}$ is convex and increasing for 
    $x \geq 1/4$).
    We thus encode $G$ by giving a
    Shannon-Fano code with parameter $p$ for the first $n/2$ rows of
    $A$; and a Shannon-Fano code with parameter $p$ for the rest of
    $A$, excluding the bits which can be deduced from the preceding
    information. Such a code has length
    \[
      |C(G)| = n_1(A) \log(1/p) + (n_0(A)-m)\log(1/(1-p))
    \]
    which results in a savings of
    \[
      s = \log(1/p_G) - |C(G)| = m\log(1/(1-p)) \ge
      \varOmega(c^2n)\log(1/(1-p)) = \varOmega(c^3) \enspace . 
    \]
    It follows that $\Pr\left\{ \mathcal{E}_3 \right\} 
    \leq 2^{-\varOmega(c^3)}$.
  \end{enumerate}
  Now the probability that $G$ contains no triangles is at most
  $\Pr \left\{\mathcal{E}_1\right\} + 
  \Pr\left\{ \mathcal{E}_2 \right\} + \Pr\left\{ \mathcal{E}_3 \right\} 
  = 2^{-\varOmega(c^3)}$. 
\end{proof}

\begin{thm}\thmlabel{triangles-down}
  If $p = c/n$ with $c > 0$, then $G \in G_{n, p}$ 
  has no triangle with
  probability at least $1 - c^3$.
\end{thm}
\begin{proof}
  Suppose that $G$ contains a triangle. We encode the adjacency matrix
  $A$ of $G$. First, we specify the triangle's vertices; and finish
  with a Shannon-Fano code with parameter $p$ for the remaining
  edges of the graph. This code has length
  \begin{align*}
    |C(G)| &= 3 \log n + (n_1(A) - 3) \log (1/p) + n_0(A) \log (1/(1 - p))\\ 
           &= \log (1/p_G) + 3 \log n - 3 \log (1/p) 
           = \log (1/p_G) + 3 \log c
           = \log (1/p_G) + \log c^3 \enspace . 
  \end{align*}
\end{proof}

Together, \thmref{triangles-up} and \thmref{triangles-down} establish
the fact that $1/n$ is a \emph{threshold function} for
triangle-freeness, \emph{i.e.}~if $p = o(1/n)$, then $G \in G_{n, p}$
has no triangle with high probability, and if $p = \omega(1/n)$, then
$G$ has a triangle with high probability. 

\section{Encoding with Kraft's Condition}
\seclabel{el}

As promised in \secref{ceilings}, we finally discuss why it has made
sense to omit ceilings in all of our encoding arguments.

Let $[0, \infty]$ denote the set of extended non-negative real
numbers, supporting the extended arithmetic operations
$a + \infty = \infty$ for all $a \in [0, \infty]$, and
$2^{-\infty} = 0$.

  Recall from \secref{nuel}
that a function $\ell : X \to [0, \infty]$ \emph{satisfies Kraft's
  condition} if
\[
  \sum_{x \in X} 2^{-\ell(x)} \leq 1 \enspace .
\]

The main observation is that neither the (Non-)Uniform Encoding Lemma
nor any of its applications has actually required the specification of an
explicit
prefix-free code: We know, by construction, that every code we have
presented is
prefix-free, but we could also deduce from Kraft's inequality that,
since our described codes satisfy Kraft's condition, a prefix-free
code with the same codeword lengths exists. Similarly, we will see that
it is actually enough to assign to every element from our universe
a codeword \emph{length} such that Kraft's condition is satisfied. These
codeword lengths need not be integers.

\begin{lem}[The encoding lemma for the Kraft inequality]\lemlabel{el}
  Let $\ell : X \to [0, \infty]$ satisfy Kraft's condition and let
  $x\in X$ be drawn randomly where $p_x > 0$ denotes the probability
  of drawing $x$.  Then
  \[
    \Pr\{ \ell(x) \le \log(1/p_x)-s\} \le 2^{-s} \enspace .
  \]
\end{lem}
\begin{proof}
  The proof is identical to that of \lemref{nuel}.
\end{proof}

The \emph{sum} of two functions $\ell : X \to [0, \infty]$ and
$\ell' : X' \to [0, \infty]$ is the function
$\ell + \ell' : X \times X' \to [0, \infty]$ defined by
$(\ell + \ell') (x, x') = \ell(x) + \ell'(x')$. Note that for any
partial codes
$C : X \nrightarrow \{0, 1\}^*, C' : X' \nrightarrow \{0, 1\}^*$, any
$x \in X$, and any $x' \in X'$,
\[
(|C| + |C'|)(x, x') = |C(x)| + |C'(x')| = |(C \cdot C')|(x, x') \enspace .
\]
In other words, the sum of the functions of codeword lengths describes
the length of codewords in concatenated codes.

\begin{lem}\lemlabel{composition-tight}
  If $\ell : X \to [0, \infty]$ and $\ell' : X' \to [0,
  \infty]$ satisfy Kraft's condition, then so does $\ell + \ell'$.
\end{lem}
\begin{proof}
  Kraft's condition still holds:
  \[
  \sum_{(x, x') \in X \times X'} 2^{-(\ell + \ell')(x, x')} = \sum_{x
    \in X} \sum_{x' \in X'} 2^{-\ell(x) - \ell'(x')} = \sum_{x \in X}
  2^{-\ell(x)} \sum_{x' \in X'} 2^{-\ell'(x')} \leq 1 \enspace
  . 
  \]
\end{proof}
This is analogous to the fact that the concatenation of prefix-free
codes is prefix-free.

\begin{lem}\lemlabel{fixed-length-tight}
  For any probability density $p : X \to (0, 1)$, the function
  $\ell : X \to [0, \infty]$ with $\ell(x) = \log (1/p_x)$ satisfies
  Kraft's condition.
\end{lem}
\begin{proof}
  \[
  \sum_{x \in X} 2^{-\ell(x)} = \sum_{x \in X} 2^{-\log (1/p_x)} =
  \sum_{x \in X} p_x = 1 \enspace . 
  \]
\end{proof}
This tells us that we can ignore the ceiling in every instance of a
fixed-length code and every instance of a Shannon-Fano code while
encoding.

We now give a tight notion corresponding to Elias codes.
\begin{thm}[Beigel~\cite{beigel:elias}]\thmlabel{elias-tight}
  Fix some $0 < \eps < e - 1$. Let $\ell : \N \to \R$ be defined as
  \[
  \ell(i) = \log i + \log \log i + \cdots + \underbrace{\log \cdots
    \log}_{\text{$\log^* i$ times}} i - (\log \log (e - \eps)) \log^*
  i + O(1) \enspace .
  \]
  Then, $\ell$ satisfies Kraft's condition. Moreover, the function
  $\ell' : \N \to \R$ with
  \[
  \ell'(i) = \log i + \log \log i + \cdots + \underbrace{\log \cdots
    \log}_{\text{$\log^* i$ times}} i - (\log \log e) \log^* i + c
  \]
  does not satisfy Kraft's condition for any choice of the constant
  $c$.
\end{thm}

It is not hard to see how \lemref{el}, \lemref{composition-tight},
\lemref{fixed-length-tight}, and \thmref{elias-tight} can be used to
give encoding arguments with real-valued codeword lengths. For
example, recall how the result of \thmref{runs-i} carried an artifact
of binary encoding. Using our new tools, we can now refine this and
recover the exact result.
\begin{customthm}{\ref*{thm:runs-i}b}
  Let $x=(x_1,\ldots,x_n)\in\{0,1\}^n$ be chosen uniformly at random
  and let $t = \ceil{\log n + s}$. Then, the probability that 
   $x$ contains
  a run of $t$ ones
  is at most $2^{-s}$.
\end{customthm}
\begin{proof}
  Let $\ell : \{0, 1\}^n \to [0, \infty]$ be such that if $x$ contains
  a run of $t = \ceil{\log n + s}$ ones, then
  $\ell(x) = \log n + n - t$, and otherwise $\ell(x) = \infty$. We
  will show that $\ell$ satisfies Kraft's condition.

  Let the function $f : \{1, \ldots, n - t + 1\} \to [0, \infty]$ have
  $f(i) = \log n$ for all $i \in \{1, \dots, n - t + 1\}$, and
  $g : \{0, 1\}^{n - t} \to [0, \infty]$ have $g(y) = n - t$ for all
  $y \in \{0, 1\}^{n - t}$. Both $f$ and $g$ satisfy Kraft's condition by
  \lemref{fixed-length-tight}. By \lemref{composition-tight}, so does
  the function
  \[
    h = f + g : \{1, \ldots, n - t + 1\} \times \{0, 1\}^{n - t} \to
    [0, \infty] \enspace ,
  \]
  where $h(i, y) = \log n + n - t$ for all $i$ and $y$. Crucially,
  each element $(i, (y_1, \dots, y_{n-1})) \in 
  \{1, \ldots, n - t + 1\} \times \{0, 1\}^{n
    - t}$ yields to an $n$-bit binary string containing a run of
  $t$ ones, namely 
  \[
  (y_1, \dots, y_{i - 1}, \underbrace{1, 1, \dots, 1}_{\text{$t$ times}},
  y_i, \dots, y_{n - t}) \enspace ,
  \]
  and this mapping is surjective. Thus, $\ell$ is as desired.
  We finish by applying \lemref{el}.
\end{proof}

\section{Summary and Conclusions}
\seclabel{summary}

We have described a simple method for producing encoding
arguments. Using our encoding lemmas, we gave original proofs for
several previously established results.

Encoding arguments often give natural and short proofs for results in
the analysis of algorithms, particularly for algorithms which directly
handle bit strings. Many of our results concern the worst-case running
time of algorithms, but encoding arguments can also be used to study
other measures of algorithmic complexity, such as the average-case
running time of algorithms, or the communication complexity of Boolean
functions~\cite{buhrman:newapps}.

Typically, to produce an encoding argument, one would invoke the
incompressibility method after developing some of the theory of
Kolmogorov complexity. Our technique requires only a basic
understanding of prefix-free codes and one simple lemma. We are also
the first to suggest a simple and tight manner of encoding using only
Kraft's condition with real-valued codeword lengths. In this light, we
posit that there is no reason to develop an encoding argument through
the incompressibility method: our Uniform Encoding Lemma is simpler,
the Non-Uniform Encoding Lemma is more general, and our technique from
\secref{el} is less wasteful. Indeed, though it would be easy to state
and prove our Non-Uniform Encoding Lemma in the setting of Kolmogorov
complexity, it seems as if the general encoding lemma from \secref{el}
only can exist in our simplified framework.

\section*{Acknowledgements}

This research was initiated in response to an invitation for the first
author to give a talk at the Ninth Annual Probability, Combinatorics
and Geometry Workshop, held April 4--11, 2014, at McGill University's
Bellairs Institute.  Many ideas that appear in the current paper were
developed during the workshop. The author is grateful to the other
workshop participants for providing a stimulating working environment.
In particular, Xing~Shi~Cai pointed out the application to runs in
binary strings (\thmref{runs-i}) and G\'abor~Lugosi stated and proved
the Non-Uniform Encoding Lemma (\lemref{nuel}) in response to the first
author's half-formed ideas about a non-uniform version of \lemref{uel},
and then later pointed us to the proof in Barron's thesis.  We would
also like to thank Yoshio Okamoto, G\"unter Rote, and the 
anonymous referees for valuable comments.

\bibliography{encoding}{}
\bibliographystyle{plain}

\end{document}